\documentclass[journal]{elsarticle}

\preprintfalse

\journal{Advances in Applied Mathematics}

\usepackage{lipsum}
\makeatletter
\def\ps@pprintTitle{%
 \let\@oddhead\@empty
 \let\@evenhead\@empty
 \def\@oddfoot{}%
 \let\@evenfoot\@oddfoot}
\makeatother

\newcommand\blfootnote[1]{%
  \begingroup
  \renewcommand\thefootnote{}\footnote{#1}%
  \addtocounter{footnote}{-1}%
  \endgroup
}

\newif\ifcomplete
\completetrue

\ifcomplete
\onecolumn
\fi

\usepackage{graphicx}
\usepackage{caption}

\usepackage{rotating}
\usepackage{amsmath,amsthm, amssymb,accents,mathalfa}
\usepackage{fancyhdr}
\usepackage{subfigure}
\usepackage{mathdots}
\usepackage{algorithm}
\usepackage{algorithmic}
\usepackage{comment}
\usepackage[final,notref,notcite]{showkeys}

\usepackage{tikz}
\usetikzlibrary{arrows,snakes,,fit,calc,shapes,backgrounds} 
\tikzstyle{box} = [draw, rectangle, rounded corners, thick, node distance=7em, text width=6em, text centered, minimum height=3.5em]
\tikzstyle{container} = [draw, rectangle, dashed, inner sep=2em]\tikzstyle{line} = [draw, thick, -latex']
\usepackage{graphics}
\usepackage{graphicx}
\usepackage[latin1]{inputenc} 
\usepackage{graphics}
\usepackage{enumerate,mathrsfs}

\usepackage{amsmath, amssymb}

\relpenalty=\maxdimen

\def\isodag{\mathrm{ISODAG}}
\def\edge{\mathrm{edge}}
\def\X{Y}
\def\Bl{\mathrm{Bl}}
\def\El{\mathrm{El}}

\newtheorem{definition}{Definition}

\newtheorem{example}{Example}
\newtheorem{problem}{Problem}
\newtheorem{remark}{Remark}

\newtheorem{theorem}{\bf Theorem}

\newtheorem{proposition}{\bf Proposition}

\newcommand{\dm}[0]{\textup{dim}}

\newcommand{\mpf}[0]{\mathfrak{p}}
\newcommand{\mq}[0]{\mathfrak{q}}

\def\ci{\perp\!\!\!\perp}

\mdseries 

 \usepackage{cancel}

\def\mdef{\stackrel{\vartriangle}{=}}
\def\mLoc{\mathrm{Loc}}
\def\mloc{\mathrm{loc}}
\def\mFac{\mathrm{Fac}}
\def\Imf{\mathrm{Im}f}
\def\Imfp{\mathrm{Im}f}

\def\dperp{\perp\!\!\!\perp}

\def\CC{\mathbb{C}}
\def\ZZ{\mathbb{Z}}
\def\nd{\mathbf{nd}}
\def\p{\mathbf{pa}}
\def\gb{Gr$\ddot{\textup{o}}$bner\,}
\def\dag{\mathrm{dag}}

\newcommand*{\mdot}{%
  \accentset{\mbox{\Large\bfseries\hspace{-.1cm} .}}\,}

\def\SN{\mathcal{N}}
\def\SNP{\mathcal{N}^{+}}

\def\EE{\mathbb{E}}

\def\eqdef{\triangleq}

\begin{document}

\begin{frontmatter}

\title{Algebraic Methods of Classifying Directed Graphical Models}

 \author[label1]{Hajir~Roozbehani}
 \address[label1]{Department of Aeronautics and Astronautics at MIT. \mbox{e-mail:~{\ttfamily hajir@mit.edu}}}
 \author[label2]{Yury~Polyanskiy}
 \address[label2]{Department
	of Electrical Engineering and Computer Science,MIT, Cambridge, MA 02139 USA. \mbox{e-mail:~{\ttfamily yp@mit.edu}}}


\address{}

\begin{abstract} 
Directed acyclic graphical models (DAGs) are often used to describe common structural properties in a family of probability distributions. This paper addresses the question of classifying DAGs up to an isomorphism. By considering Gaussian densities, the question reduces to verifying equality of certain algebraic varieties.   
A question of computing equations for these varieties has been previously raised in the literature. Here it is shown
that the most natural method adds spurious components with singular principal minors, proving a conjecture of Sullivant.
This characterization is used to establish an algebraic criterion for isomorphism, and to provide a randomized algorithm for
checking that criterion. Results are applied to produce a list of the isomorphism classes of tree models 
on 4,5, and 6 nodes. Finally, some evidence is provided to show that projectivized DAG varieties contain useful information in the sense that their relative embedding is closely related to efficient inference. 
\end{abstract}


\end{frontmatter}
\blfootnote{The research was supported by the Center for Science of Information (CSoI),
an NSF Science and Technology Center, under grant agreement CCF-09-39370 and the NSF 
CAREER award under grant agreement CCF-12-53205.}

%

%

\section{Introduction}

Consider two directed graphical models (or directed acyclic graphs, DAGs) on random variables $(A,B,C)$:
\begin{equation}\label{eq:thr}
		A \to B \to C \qquad B \leftarrow A \to C   
\end{equation}
(See~\cite{lau96} for background on graphical models.) In this paper, we will say that these two models
are \textit{isomorphic} (as graphical models). Roughly, this means that after relabeling ($A\leftrightarrow B$), the two resulting models describe the same collection of joint distributions
$P_{A,B,C}$. Note that the so defined isomorphism notion is weaker than the (directed) graph isomorphism: the
graphs in~\eqref{eq:thr} are not isomorphic.  

On
the other hand, there does not exist any relabeling making~\eqref{eq:thr} equivalent to 
\begin{equation}\label{eq:thr2}
		B \to C \leftarrow A 
\end{equation}
In fact, a simple
exercise in $d$-separation criterion shows that~\eqref{eq:thr} and~\eqref{eq:thr2} list all possible isomorphism classes
of directed \textit{tree} models
on three variables.  However, note that the above DAGs are all isomorphic as undirected graphs. 

The goal of this paper is to provide (computational) answer to:
\textit{What are the isomorphism classes of directed graphical models on $n$ nodes?}

Note that when variables $(A,B,C)$ are jointly Gaussian and zero-mean, then conditions such as~\eqref{eq:thr} can be
stated as algebraic constraints on the covariance matrix:
\begin{equation}\label{eq:thr3}
		\EE[AB] \EE[BC] = \EE[AC] \EE[B^2]. 
\end{equation}
This suggests that checking isomorphism of models can be carried out via algebraic methods. Indeed, one needs to recall (see \cite{lev01}) that graphical models equality can be tested by restricting to Gaussian random variables.

In this paper, we associate with every DAG two subsets of covariance matrices:
\begin{itemize}
	\item all non-singular covariance matrices satisfying DAG constraints (denoted $\mloc(G)\cap\Sigma^{++}$ below)
	\item all covariance matrices satisfying DAG constraints (denoted $\mloc(G)$ below)
\end{itemize}
We give an \textit{analytic} result: while $\mloc(G)$ is not necessarily (Euclidean) closed, closures of both sets
coincide. 

Next, we switch to the \textit{algebraic} part. Due to the analytic fact above, much simpler equations for non-singular
matrices can be used to completely characterize the Zariski closure of
$\mloc(G)$ (denoted
$X_G$ below). 
Aesthetically pleasing is the fact that $X_G$ is always an irreducible complex variety (affine and rational).
Furthermore, two graphical models $G$ and $G'$ define the same set of conditional independence constraints if and only
if $X_G = X_{G'}$.

For large graphs it is important to reduce the number of equations needed to describe $X_G$. 
The natural set of
equations (denoted $I_G$ below) turns out to be too small: its solution set $V(I_G)$ contains $X_G$ and a
number of spurious components. We show how to get rid of these spurious components, proving that
$$ X_G = V((I_G:\theta_0^m))\,, $$
where $\theta_0$ is an explicit polynomial (and establishing Conjecture 3.3 of Sullivant~\cite{sul08}). This provides a convenient method for computing $X_G$. After these preparations, we
give our main result: isomorphism question $G\stackrel{?}{\sim} G'$ is equivalent to comparing intersections of $X_G$
and $X_{G'}$ with a certain invariant variety. We give a randomized algorithm for this and apply it to provide a list
of isomorphism classes on 4,5, and 6 nodes. 

The question of DAG isomorphism does not seem to have appeared elsewhere, though the closely related question of DAG equivalence (or Markov equivalence \cite{and97}) is well-studied. As mentioned in \cite{chi02}, the natural space to work with when doing model selection or averaging over DAGs is that of their equivalence classes. In practice, the number of DAGs in an equivalence class encountered during model selection can be large. Some examples are presented in Table 4 of \cite{chi02} where, for instance, the learning algorithm discovers equivalence classes on 402 nodes with more than $7\times 10^{21}$ members. However, one should keep in mind that not all equivalence classes are as large. In fact, Steinsky \cite{ste03} showed, by a recursive method, that roughly 7 percent of equivalence classes of graphs with 500 nodes or less consist only of one element. Recent results \cite{wag13} indicate that this ratio is valid asymptotically as well. Such classes appear to be the most common type of equivalences classes of DAGs \cite{gil01}. In general, it is expected, based on observations on small graphs, that the ratio of equivalence classes to DAGs be around 0.27 \cite{gil01}\cite{pen13}. Nevertheless, some equivalence classes (such as those that appear in \cite{chi02}) are quite large. This has motivated the need to represent DAGs, and among the representatives that are relevant in this regard are the essential graphs{\footnote{The essential graphs were originally known as {completed patterns} and were introduced in \cite{ver90} as the maximal invariants associated to equivalence classes of DAGs. They were also studied in \cite{mee95} under the name {maximally oriented graphs}.}} \cite{and97} and the characteristic imsets \cite{std10}. Both these methods have a combinatorial flavor and this work provides an algebraic alternative. The word algebraic here means commutative-algebraic, unlike in~\cite{std10}. We remark that the two mentioned methods can also be applied to solve the isomorphism problem. For instance, using the results in \cite{and97} one can reduce DAG isomorphism to the isomorphism of certain directed multi-graphs. 
In fact, this gives a sense of the inherent computational difficulty involved in working with the isomorphism class of a DAG.

While the notion of Markov equivalence makes sense in the setting of \cite{chi02}, there are situations where it is natural to want to work with the isomorphism class of a DAG-- the Markov equivalence class modulo permutations of variables. For instance, the recent results in \cite{kwi10} imply that there is a precise sense in which the isomorphism classes of all large graphs that admit efficient inference are related to graphs that look like the (unlabeled) trees listed in Figure~\ref{fig_classes}, but are far from the complete DAGs. 
An exact description of such models, however, is problematic by the subsequent results in \cite{kwi10}. It thus appears reasonable to find good ways to approximate them, and for that we resort to the family of projective varieties. 

It is also important to mention that the idea of associating an algebraic variety to a conditional independence (CI)
model has been previously explored in a number of publications, among which we will
discuss~\cite{lnv07,mat95,mat95p,mat99,sim06p,sim06,drt08,drt09,mat97,gar05}. Some of our preparatory propositions can
be found in the literature in slightly weaker forms and we attempt to give references.  The main novelties are:
\begin{itemize} 
\item We essentially leverage the directed-graph structure of the model (as opposed to general CI models) to infer
stronger algebraic claims. In particular, our treatment is base independent -- although for readability we present
results for the varieties over $\CC$.
\item We present a computational procedure for answering the isomorphism question.
\end{itemize}


\subsection{Preliminaries}
Directed acyclic graphical models are constraints imposed on a set of probability distributions:
\begin{definition}
A directed acyclic graphical model (DAG) ${G}/k$ is the data:
\begin{itemize}
\item A set of indices $[n]:=\{1,\cdots,n\}$ that are nodes of a directed acyclic graph. We frequently assume the nodes
to be topologically sorted, i.e., $i<j$ whenever there
is a path in the graph from $i$ to $j$.
\item A list $\mathcal{M}_{{G}}$ of {\em imposed} (a.k.a local Markov) relations
\begin{equation}\label{eq:topoall}
		i\ci \nd(i)| \p(i) \quad  
\end{equation}
where $\p(i)$ denotes the set of parents of $i\in [n]$ and $\nd(i)$ is the set of non-descendants of $i$ in the directed graph. 
\item A subset $\mathcal{M}_{{G}}^{topo}$ of topologically sorted local Markov relations:
\begin{equation}\label{eq:topomark}
i\ci \nd(i) \cap \{j: j<i\}| \p(i) \quad 
\end{equation}
\item A set of {\em ${G}$-compatible} joint probability distributions
\[
\mLoc({G}):=\{P_X | \,I\ci J|K\in \mathcal{M}_{G}\Rightarrow X_I\ci X_J|X_K \},
\]
where $X$ is a $k^n$-valued random variable\footnote{We mostly work with $k=\mathbb{R}$ or $\CC$. Since $\mathbb{R}$ and $\CC$ are measurably isomorphic, it does not matter which one we pick. We write ${G}/k$ if we need to emphasize the base field $k$. }. 
\item A set of {\em implied} relations
\[
\mathcal{C}_{{G}}:=\cap_{P_X\in \mLoc({G})}\{I\ci J|K \,\,\mathrm{s.t}\,\, X_I\ci X_J|X_K\}.
\]
\end{itemize}
\end{definition}

Given a collection of such models, it is often of interest to find representatives for their isomorphism classes (see also \cite{lnv07,mat95})-- these are models that have the same compatible distributions modulo labelings of variables: 
\begin{definition} Let $\mathcal{Q}$ be a permutation invariant family of distributions. Two DAGs ${G},{G}'$ are called {$\mathcal{Q}$-equivalent}
if 
\[
\mLoc({{G}})\cap \mathcal{Q}=\mLoc({{G}'})\cap \mathcal{Q}.
\]
When $\mathcal{Q}$ is the set of all distributions, we call such models equivalent. Likewise, two DAGs ${G},{G}'$ are called {$\mathcal{Q}$-isomorphic} 
if 
\[
p_{X_1\cdots X_n}\in \mLoc({{G}})\cap \mathcal{Q}\iff p_{X_{\pi(1)}\cdots X_{\pi(n)}}\in \mLoc({{G}'})\cap \mathcal{Q}
\]
for some permutation $\pi$ of indices. When $\mathcal{Q}$ is the set of all distributions, we call such models isomorphic and denote the isomorphism class of $G$ by $[G]$. 
\label{def_equiv}
\end{definition}

We shall mainly focus on characterizing isomorphism classes of DAGs. A related question is that of understanding the
structure of conditional independence constraints -- see for the case of discrete random variables \cite{mat95p,mat99,mat95}, positive discrete random variables \cite{sim06p}, non-singular Gaussians \cite{lnv07},  and general Gaussians \cite{sim06}. 

Let $H=H_1\times\cdots\times H_n$ be a product measure space endowed with the $\sigma$-algebra $\mathcal{H}=\mathcal{H}_1\otimes\cdots\otimes\mathcal{H}_n$. We assume that $H_i$ is measurably isomorphic to $\mathbb{R}$ and that $\mathcal{H}_i$ is a Borel $\sigma$-algebra for all $i$. The next property, factorization, 
relies on a digraph structure and pertains only to DAGs:
\begin{definition}
A probability measure $P$ defined on $(H,\mathcal{H})$ is said to factorize w.r.t a DAG if it can be written as
\[
P(\mathcal{A})=\int_{\mathcal{A}}\prod_{i}^n K_{i|\mathrm{pa}(i)}(dx_i|x_{\mathrm{pa}(i)}) \quad \forall \mathcal{A}\in \mathcal{H},
\]
where $K_{i|\mathrm{pa}(i)}$'s are conditional probability kernels (which exist by \cite[Theorem 2.7]{cin11}) 
and $K_{i|\mathrm{pa}(i)}(dx_i|x_{\mathrm{pa}(i)})=\mu_i(dx_i)$ if $i$ has no parents in $G$.
\label{def_fac}
\end{definition}

Given a DAG $G$, we denote by $\mFac(G)$ the set of distributions that factorize w.r.t $G$. It is known (c.f. Section II.A) that 
\[
\mFac(G)=\mLoc(G).
\]
This means that two DAGs are equal (isomorphic) in the above sense if and only if they factorize the same set of distributions (modulo the labeling of the variables). 

\subsection{Notation}
\begin{itemize}
\item $\SN$ is the set of real valued Gaussians
\item $\SNP$ is the non-singular subset of $\SN$. 
\item $\Sigma=[\sigma_{ij}]$ is the affine space $\CC^{{n+1 \choose 2}}$ of Hermitian $n\times n$ matrices.
\item $\Sigma^+$ is the positive semi-definite (PSD) subset of $\Sigma$.
\item $\Sigma^{++}$ is the positive definite (PD) subset of $\Sigma$.
\item $\Sigma^{\mdot{}\,}$ is the subset of matrices in $\Sigma$ with non-zero principal minors\footnote{Note that $\Sigma^{\mdot{}\,}$ is Zariski open, while $\Sigma^{+},\Sigma^{++}$ are described by inequalities. }. 
\item $\hat{\Sigma}$ is the subset of $\Sigma$ consisting of matrices with ones along the diagonal. We also set $\hat{\Sigma}^+\mdef\hat{\Sigma}\cap \Sigma^+$, $\hat{\Sigma}^{++}\mdef\hat{\Sigma}\cap \Sigma^{++}$, and $\hat{\Sigma}^{\mdot{}\,}\mdef\hat{\Sigma}\cap \Sigma^{\mdot{}\,}$. 
\item $\mloc(G)$ is the set of covariance matrices in $\mLoc(G)\cap \SN$. 
\item $f_G$ is the rational parametrization defined in II.B.
\item Given $S\subset \Sigma$, $[S]_{}$ and $[S]_Z$ are its standard and Zariski closures\footnote{The closure is always taken inside the affine complex space.}, respectively.
\item Given $S\subset \Sigma$, $I(S)$ is the ideal of polynomials that vanish on $[S]_Z$. 
\item Given an ideal $I$, the associated algebraic set is given by
\[
{V}(I)=\{ x\in \CC^n| f(x)=0 \quad \forall f\in I\}.
\]
\item $X_G\mdef [\mloc(G)]_Z$, $\mpf_G\mdef I(X_G)$, $\hat{X}_G\mdef [\mloc(G)\cap \hat{\Sigma}]_Z$.
\item $\X_G$ is the closure of $\hat{X}_G$ inside $\mathbb{P}^{n \choose 2}$. 
\end{itemize}

\subsection{Overview of main results} 

Our main purpose is to show that the computational tools in algebra are relevant for addressing the following problem: 

\begin{problem}
Given two DAGs, determine if they are isomorphic.
\end{problem}

Our starting point is to show that isomorphism and $\mathcal{N}^{+}$-isomorphism are equivalent for DAGs (see Section II.C). It is well known that checking $\mathcal{N}^{+}$-equivalence reduces to checking equality of algebraic subsets inside the positive definite cone (see for instance \cite{drt08,drt09,mat97,sul08}). 
This follows from the next proposition: 
\begin{proposition}[Lemma 2.8 in \cite{std05}]
Let $X \sim N(\mu,\sigma)$ be an $m$-dimensional Gaussian vector and $A,B,C \subset [m]$ be pairwise disjoint index sets. Then $X_A\ci X_B | X_C$ if and only if the submatrix $\sigma_{AC,BC}$ has rank equal to the rank of $\sigma_{CC}$. Moreover, $X_A\ci X_B | X_C$ if and only if $X_a\ci X_b | X_C$ for all $a \in A$ and $b \in B$.
\label{pro_minors}
\end{proposition}
\begin{remark}
Note that the rank constraint is equivalent to vanishing of the minor $|\sigma_{AC',BC'}|$ for a maximal $C'\subset C$
such that $X_{C'}$ is non-singular  \footnote{A vector random variable is said to be non-singular if its distribution
admits a density w.r.t. product Lebesgue measure.}.
\end{remark}

Proposition \ref{pro_minors} enables us to think algebraically and/or geometrically when deciding Gaussian equivalence. Indeed, 
it states that $\mloc(G)\cap \Sigma^{++}$ can be identified with the positive definite subset of the real solutions to the polynomial equations generated by the implied relations in $G$. Working with such subsets, however, is not convenient from a computational point of view. This motivates the next problem:
\begin{problem}
Give an algebraic description of $\mloc(G)$.
\end{problem}  

Let ${J}_{G}$ be the ideal generated by the minors $|\sigma_{iK,jK}|$ of the implied relations $i\ci j|K\in \mathcal{C}_G$ inside $\CC[\Sigma]$. Similarly, the minors of imposed relations of $G$ generate an ideal $I_G\subset J_G$ in $\CC[\Sigma]$. Note that this ideal coincides with that generated by the toposorted imposed relations. The corresponding ideals generated inside $\CC[\hat{\Sigma}]$ are denoted by $\hat{I}_G,\hat{J}_G$.
With the established notation, for example, Proposition~\ref{pro_minors} implies
\begin{equation}\label{eq:vcmloc}
		{V}(I_G) \cap \Sigma^{++}\cap \mathbb{R}^{n+1\choose 2} = \mloc(G) \cap \Sigma^{++}. 
\end{equation}
We address the above problem by identifying $X_G$ with an irreducible component of ${V}(I_G)$. 
It is a curious fact that the points in ${V}(I_G)\cap \hat{\Sigma}^{++}$
correspond to covariances of circularly symmetric Gaussians that satisfy the CI constraints of $G/\CC$. Thus if we work
with complex Gaussians, we may avoid intersecting with the reals in (\ref{eq:vcmloc}).

In Section II, we first prove some geometric results, which can be summarized in the following diagram
\begin{center}
\tikzset{node distance=1.5cm, auto}
\begin{tikzpicture}
  \node (D0) {$\Imfp_G\cap \Sigma^{+}$};
  \node (D1) [right of=D0,node distance=1.75cm]{$=\mloc(G)\subset$};
  \node (D2) [right of=D1,node distance=2.25cm]{$[\mloc(G)\cap \Sigma^{++}]_{}\subset$};
  \node (D3) [right of=D2, node distance=2.2cm]{$X_G$};
  \node (C0) [below of =D0, node distance=.5cm] {\rotatebox{90}{$\subsetneq$}};
  \node (C1) [below of =D1, node distance=.5cm] {\rotatebox{90}{$\supsetneq$}};
  \node (C3) [below of =D3, node distance=.5cm] {\rotatebox{90}{=}};
  \node (B0) [below of =C0, node distance=.5cm] {$\mloc(G)\cap \Sigma^{++}$};
  \node (B1) [below of =C1, node distance=.5cm] {$V(I_G)$};
  \node (B3) [below of =C3, node distance=.5cm] {$[V(I_G)\cap \Sigma^{\mdot}\,]_Z$};
\end{tikzpicture}
\end{center}
The same inclusions hold if we replace $(I_G,\Sigma)$ with $(\hat{I}_G,\hat{\Sigma})$, $\Sigma^{\mdot}\,$ with
$\Sigma^{++}$, or $I_G$ with $J_G$.  

It is known that $[\mloc(G)\cap \Sigma^{++}]_Z$ is a complex irreducible rational algebraic variety, cf.~\cite{sul08}. Here we further show that it coincides with $X_G$ and characterize $\mathfrak{p}_G=I(X_G)$ in two different ways: as the saturated ideal of $I_G$ at $\theta_0=\prod_{A\subset [n]}(|{\Sigma}_{AA}|)$ (Conjecture 3.3 in \cite{sul08}), and as the unique minimal prime of $I_G$ contained in the maximal ideal $\mathfrak{m}_I$ at the identity. We thus have the following relations inside $\CC[{\Sigma}]$:
\begin{align*}
&I_G\subset J_G\subset S^{-1}J_G\cap {\CC}[\Sigma]= S^{-1}I_G\cap \CC[\Sigma]\\
&=I(\mloc(G)\cap\Sigma^{++})=\mathfrak{p}_G\subset \mathfrak{m}_I,
\end{align*}
where $S=\{\theta_0^n|n>0\}$. One can replace $(\mloc(G),{\Sigma})$ with $(\mloc(G)\cap \hat{\Sigma},\hat{\Sigma})$ in the above. We note that the above relations hold verbatim over $\ZZ[\Sigma]$ and other base rings. 
Our main statement, shown in \ref{sec:algrep}, is that two DAGs $G,G'$ are isomorphic if and only if
\[
S^{-1}I_G\cap \CC[\Sigma]^\Pi=S^{-1}I_{G'}\cap \CC[\Sigma]^\Pi.  
\]

We use the above results to provide a randomized algorithm for testing DAG isomorphism in \ref{sec:rand}. 
\ifcomplete Section III concerns the special case of directed tree models. In particular, we show that 
$\hat{I}_T$ is a prime ideal for a tree model $T$ and hence $\hat{I}_T=I(\mloc(T)\cap \hat{\Sigma})$. This is analogous to primality of $J_T$, the ideal of implied relations, shown in \cite{sul08} (see Corollary 2.4 and Theorem 5.8). 
We thus have that two directed tree models $T$ and $T'$ are equal if and only if $\hat{I}_T=\hat{I}_{T'}$. Moreover, we use our randomized algorithm to list the isomorphism classes of directed tree models for $n=4,5,$ and $6$ nodes. \fi  
The number of isomorphism classes of directed tree models  found by our procedure is $1,1,2,5,14,42,142,...$ for $n\ge 1$. Curiously, the first 6 numbers are Catalan but the 7th is not. 

\section{Main results}
\subsection{Factorization and local Markov properties}
In this section we show that isomorphic DAGs factorize the same set of probability distributions modulo the labeling of the variables. A theorem of Lauritzon (see \cite[Theorem 3.27]{lau96}) says that a non-singular measure satisfies the local Markov property if and only if its density factorizes. Let $(H,\mathcal{H})$ be as in Definition \ref{def_fac}. One can further state:
\begin{proposition} Let ${G}$ be a DAG and $P$ a probability measure defined on $(H,\mathcal{H})$. The following are
	equivalent:
	\begin{enumerate}
		\item $P$ factorizes w.r.t to ${G}$.
		\item $P$ satisfies all imposed constraints~\eqref{eq:topoall} w.r.t ${G}$.
		\item $P$ satisfies topologically sorted constraints~\eqref{eq:topomark} w.r.t. $G$.
	\end{enumerate}  
	In particular, $\mFac(G)=\mLoc(G)$.
\label{pro_locfac}
\end{proposition}
\begin{proof}
\ifcomplete
$1\implies 2 \implies 3$ are obvious.  We show $3\implies 1$. The main step is to show that if $1\ci 3| 2$, then there exists a conditional probability kernel $K_{3|2}$ such that for all $\mathcal{A}\in \mathcal{H}_3$ we have
\[
\int K_{3|12}(\mathcal{A}|x_1,x_2)dP_{X_1X_2}=\int K_{3|2}(\mathcal{A}|x_1)dP_{X_1X_2}.
\]
The general result then follows from this by induction. Note that the inductive step requires $P$ to satisfy only the toposorted constraints of $G$. Without loss of generality, let $K_{3|12}$ and $K_{1|2}$ be regular branches of conditional probabilities. 
Set
\[
K_{3|2}(\mathcal{A}|x_2):=\int K_{3|12}(\mathcal{A}|x_1,x_2)K_{1|2}(dx_1|x_2).
\]
We want to prove that $K_{3|2}$ is a regular branch of conditional probabilities $P_{X_3|X_2}$. Clearly,  $K_{3|2}(\cdot|x_2)$ is a probability measure for all $x_2$. Now fix $\mathcal{A}\in \mathcal{H}_3$. By \cite[Theorem I.6.3]{cin11},  $K_{3|2}(\mathcal{A}|\cdot)$ is measurable as well, hence,  it is a regular branch of conditional probabilities.

Claim: $K_{3|2}=K_{3|12}$. Suppose this is not the case. Then there exists $\epsilon>0$ such that
\[
\mathcal{A}=\{K_{3|12}(\mathcal{L}|x_1,x_2)>K_{3|2}(\mathcal{L}|x_2)+\epsilon\}
\]
has non zero probability for some $\mathcal{L}\in \mathcal{H}_3$. Let $\mathcal{F}_I$ be the $\sigma$-algebra generated by $X_I$. By the local Markov property
\begin{align*}
\mathbb{E}[1_\mathcal{L}1_\mathcal{A}]&=\mathbb{E}[\mathbb{E}_{\mathcal{F}_{1,2}}[1_\mathcal{L}1_\mathcal{A}]]=\mathbb{E}[1_\mathcal{A}\mathbb{E}_{\mathcal{F}_{2}}[1_\mathcal{L}]]\\
&=\int_{ \mathcal{A}} K_{3|2}(\mathcal{L}|x_2)dP_{X_1X_2}.
\end{align*}
Now by direct computation
\begin{align*}
\mathbb{E}[1_\mathcal{L}1_\mathcal{A}]&=\int_{\mathcal{A}}K_{3|12}(\mathcal{L}|x_1,x_2)dP_{X_1X_2} \\
&>\int_{\mathcal{A}} K_{3|2}(\mathcal{L}|x_2)dP_{X_1X_2}+\epsilon \mathbb{P}[\mathcal{A}]\\
&=\mathbb{E}[1_\mathcal{L}1_\mathcal{A}]+\epsilon \mathbb{P}[\mathcal{A}],
\end{align*}
which is a contradiction. This completes the proof.
\fi
\end{proof}

\subsection{Weak limits of factorable Gaussians}

This section provides a characterization of the singular distributions in $\mloc(G)$ as the weak limit of sequences in $\mloc(G)\cap \Sigma^{++}$. Note that since $(H,\mathcal{H})$ is a topological space, weak convergence $P_{X_n}\stackrel{w}\rightarrow P_X$ is well-defined .

\ifcomplete We note that in the case of Gaussians $X_n\sim N(0,\sigma_n), X\sim N(0,\sigma)$, $P_{X_n}\stackrel{w}{\rightarrow} P_X$ is equivalent to $\sigma_n\rightarrow \sigma$ in the standard metric.\else \fi
To characterize $\mloc(G)\cap \partial \mathcal{N}^+$, we shall find it useful to work with the parametrization 
\begin{equation}
X_i=\sum_{j<i} \alpha_{ij} X_j+\omega_iZ_i,
\label{eq_rational}
\end{equation} 
where $Z_i$'s are independent standard Gaussians. Suppose that  $\alpha_{ij}= 0$ for all $(i,j)\notin E$, where $E$ denotes the set of (directed) edges of $G$.
Then this parametrization gives a polynomial map $f_G:\mathbb{R}^{|E|+n}\mapsto \mathbb{R}^{n+1 \choose 2}$, sending $\{\alpha_{ij},\omega_i\}$ to $\mathrm{cov}(X)$. Indeed, starting from (\ref{eq_rational}), one can write
\[
\sigma_{ik}=\sum_{j<i}\alpha_{ij} \sigma_{jk}+\omega_i\gamma_{ik}
\]
where $\gamma_{ik}=\mathrm{Cov}(Z_i,X_k)$, $\sigma_{ik}=\mathrm{Cov}(X_i,X_k)$. Note that $\gamma_{ik}=0$ for $k<i$. With this notation, we can write
\begin{align*}
\gamma_{ki}&=\sum_{j>i}\alpha_{ij}\gamma_{kj}+\omega_i\delta_{ik}=\sum_{j>i}\gamma_{kj}\alpha^*_{ji}+\omega_i\delta_{ik}.
\end{align*}
Set $\Gamma:=[\gamma_{ij}],A:=[\alpha_{ij}],\Omega:=[\omega_{ii}],\Sigma:=[\sigma_{ij}]$. We can write the above equations in matrix form:
\[
\Sigma=A\Sigma+\Omega\Gamma, \quad \Gamma=\Gamma A^*+\Omega.
\]
Hence,
\[
\Sigma=(I-A)^{-1}\Omega^2(I-A^*)^{-1}.
\]
The image of $f_G$ is Zariski dense in $[\mloc(G)\cap \Sigma^{++}]_Z$:
\begin{proposition}[\it Proposition 2.5 in \cite{sul08}]\label{prop:irred}
Let $G$ be a DAG and $E$ be its set of edges. Then $[\mloc(G)\cap \Sigma^{++}]_Z$ is a rational affine irreducible variety of dimension $n+|E|$.
\end{proposition}

The next Proposition shows that $$X_G=[\mloc(G)\cap \Sigma^{++}]_Z.$$

\begin{proposition} Let $G$ be a DAG. Then
\begin{enumerate}[(a)]
\item $\mloc({G})\cap \Sigma^{++}$ is dense in $\mloc({G})$.
\item $\mloc(G)\cap \hat{\Sigma}^{++}$ is dense in $\mloc(G)\cap \hat{\Sigma}$. 
 \end{enumerate}
\label{pro_dense}
\end{proposition}  

\begin{proof}
\ifcomplete
For part (a), we want to show $\mloc(G)\subset [\Imf_G\cap \Sigma^{++}]$. Given a Gaussian $X\in \mLoc(G)$, start with a representation
\[
X_i=\sum_{j<i} \alpha_{ij} X_j+\omega_iZ_i,
\]
where $Z_{i}$'s are i.i.d. Gaussians. We need to show that there exists a $G$-compatible representation $\{\alpha_{ij}',\omega_i\}$ for $X$, i.e., $\alpha'_{ij}= 0$ for all adjacent nodes $i,j$ in $G$. We may assume by induction that the $\alpha_{ij}$'s are $G$-compatible for $i,j<n$. Now write
\[
X_n=X_{\mathbf{pa}}+X_{\mathbf{npa}}+\omega_nZ_n,
\]
where $$X_{\mathbf{pa}}=\sum_{i\in \mathbf{pa}(n)} \alpha_{in} X_i, \,\,X_{\mathbf{npa}}=\sum_{i\in ([n-1]\backslash\mathbf{pa}(n)} \alpha_{in} X_i.$$
 Note that, for general random variables $A,B,C$, we have
\[
A\ci (B+C) |C\iff A\ci B|C.
\]
It thus follows that
\[
X_{\mathbf{n}\p(n)}\ci X_{\mathbf{n}\p(n)}+\omega_nZ_n|X_{\p(n)}.
\]
Now observe that for independent random variables $A,Z$
\begin{align*}
A\ci A+Z \implies \mathbb{E}[(A+Z)A]=\mathbb{E}[A+Z]\mathbb{E}[A],
\end{align*}
which implies $\mathbb{E}[A^2]=\mathbb{E}[A]^2$. In particular, if $A$ is a Gaussian then it must be a constant. It follows from this observation that $X_{\mathbf{n}\p(n)}$ is a linear function of $X_{\p(n)}$, say $X_{\mathbf{n}\p(n)}=cX_{\p(n)}$. The $G$-compatible $\alpha'_{in}$'s are then obtained by setting $\alpha'_{in}=(1+c)\alpha_{in}$ if $i\in \p(n)$ and $\alpha'_{in}=0$ otherwise.

For part (b), given $\sigma\in \mloc(G)\cap \hat{\Sigma}^{++}$, we can find a sequence $\sigma_n$ in $\mloc(G)\cap \Sigma^{++}$ that converges to $\sigma$. Pass to a subsequence with $\sigma_{ii}\neq 0$ for all $i$. Normalize the coordinates (using the defining equations) along $\sigma_{ii}$'s to obtain a sequence in $\mloc(G)\cap \hat{\Sigma}^{++}$. 
Normalization is continuous around $\sigma$. Thus, the new sequence convergences to $\sigma$ as well.
\fi
\end{proof}

\ifcomplete
\begin{remark} The above proof also shows that \mbox{$\mloc(G)=\Imfp_G\cap \Sigma^{+}$}.
\end{remark}
\fi
This proposition shows that $X_G$ contains all $G$-factorable Gaussians. There are, however, (singular) covariances on $X_G$ that are not $G$-compatible. In other words, unlike independence, conditional independence is not preserved under weak limits as shown in the following example.
\begin{example}[$\mloc(G)$ is not closed] 
Let $X_n \sim N(0,1)$, $W_n\sim N(0,1)$ be independent Gaussians. Set $Z_n=X_n$ and $Y_n=\frac{1}{n} X_n+\frac{\sqrt{n^2-1}}{n}W_n$. Then $X_n\ci Z_n|Y_n,W_n$ for all $n$ and $P_{X_n,Y_n,Z_n,W_n}\stackrel{w}{\rightarrow} P_{X,Y,X,W}$ with $X\sim N(0,1),W\sim N(0,1)$ and $Y=W$. However, $$X\not\ci X|W.$$ 
Thus the closure of $\mloc(G)\cap \hat{\Sigma}^{++}$ strictly contains \mbox{$\mloc(G)\cap \hat{\Sigma}$}. 
\label{ex_cicounter}
\end{example}

\begin{remark}
\indent
 In general, the weak convergence of the joint $P_{X^{(n)}}\stackrel{w}{\rightarrow} P_{X}$ does not imply that of the conditional kernels $P_{X^{(n)}_i|X^{(n)}_j}\stackrel{w}{\rightarrow} P_{X_i|X_j}$. If the latter conditions are also satisfied, then conditional independence is preserved at the (weak) limit\footnote{This follows directly from the lower semi-continuity of divergence.}. 
\end{remark}

\ifcomplete
\subsection{DAG isomorphism}
The next result states that isomorphism of DAGs can be decided inside $\mathcal{N}^+$\footnote{This statement generalizes to the class of chain graphs. See \cite{lev01} for details.}:
\begin{proposition}[Theorem 5.1 in \cite{lev01}]
Let $G,G'$ be DAGs. Then the following are equivalent:
\begin{enumerate}[(a)]
\item $G$ and $G'$ are equal.
\item $G$ and $G'$ are $\mathcal{N}$-equal.
\item $G$ and $G'$ are $\mathcal{N}^+$-equal.
\end{enumerate}
\label{pro_dagiso}
\end{proposition}
This property is also known as the faithfulness of Gaussians in the statistics literature (cf.~ \cite{uhl13}). 
Let us point out that, in general, $\SNP$-isomorphic models are not $\SN$-isomorphic as shown in the next example.

\begin{example}
Consider the models
\begin{align*}
{G}_1&: 1\ci 3|2\quad\&\quad 1\ci 2|3\quad\&\quad 2\ci3|1\\
{G}_2&: 1\ci 2\ci 3
\end{align*}
A non-singular Gaussian belongs to the first model if and only if it belongs to the second model. However, a Gaussian $X_1=X_2=X_3$ is only compatible with the first model.
\label{ex_nplus}
\end{example}

\fi

\subsection{DAG varieties and ideals}
Here we provide some algebraic and geometric descriptions for $\mloc(G)$:
\begin{theorem} Let $G$ be a DAG and let $\theta_0=\prod_{A\subset [n]}(|{\Sigma}_{AA}|)$.
\noindent
\begin{enumerate}[(a)]
\item There is a Zariski closed subset $B_G$ so that 
$$ V(I_G) = X_G \cup B_G $$
where $B_G \subseteq V(\theta_0) = \{\theta_0=0\}$.
\item Let $\mpf_G=I(\mloc(G))$ so that $X_G=V(\mpf_G)$. Then $\mpf_G$ is a prime ideal obtained by saturating $I_G$
\begin{equation}\label{eq:mpfg}
		\mpf_G = S^{-1} I_G \cap \CC[\Sigma] 
\end{equation}
at the multiplicatively closed set $S=\{\theta_0^n, n=1,\ldots\}$.
\item $X_G$ is smooth inside $\Sigma^{\mdot}$.
\end{enumerate}
\label{pro_real2}
\label{pro_embcomp}
\label{pro_pdecom}
\label{cor_conj}
\label{thm:thm1}
\end{theorem}
\begin{remark}
\noindent
\begin{enumerate} [(a)]
\item In Theorem \ref{pro_pdecom}b, we can replace $I_G$ with $J_G$. 
\item It follows that $V(I_G)$ and $V(J_G)$ do not miss a single $G$-compatible Gaussian, but can add some bad components to the boundary $\partial \Sigma^{++}$. Theorem \ref{cor_conj}b states this in algebraic terms and provides a proof of Conjecture 3.3 in \cite{sul08}. Theorem 8 in \cite{gar05} gives an analogous result for the implied ideals of discrete random variables. 
\item In Theorem \ref{cor_conj}b, one can replace $\theta_0$ with the product of principal minors $|\sigma_{KK}|$ where $K$ appears as a conditional set in some imposed relation $i\ci j|K$.
\item There are many equivalent ways to recover $\mpf_G$ from $I_G$ besides~\eqref{eq:mpfg}. Indeed, (e.g.~\cite[Chapter
4]{ati69}) we have
$$ \mpf_G = (I_G:\theta_0^m)$$
for all $m$ sufficiently large. Another characterization is from primary decomposition of $I_G$:
$$ I_G = \mpf_G \cap \mq_1 \cdots \mq_r\,,$$
where $\mpf_G$ is the unique component that is contained in maximal ideal $\mathfrak{m}_x$ corresponding to covariance
matrix $x$ with non-singular principal minors (e.g. identity).
\ifcomplete
\item One can ask if Theorem 1 generalizes, i.e., if $[V(I_{G})\cap \Sigma^{++}]_Z=[V(I_{G})\cap \Sigma^{\mdot\,}]_Z$ for any conditional independence model ${G}$. This is equivalent to asking if  some component of $V(I_{G})$ can intersect $\Sigma^{\mdot\,}$ but avoid $\Sigma^{++}$. This question appears in \cite{mat05}, and remains open so far as we know. 
\fi
\end{enumerate}
\end{remark}

\begin{proof}
\ifcomplete
Let us consider the ring $S^{-1} \CC[\Sigma]$ and 
if $ij$ is an edge in $G$ or if $i=j$, call $\sigma_{ij}$ an edge variable and the rest are non-edge variables. Denote by $\Sigma_{\edge}$ the
subspace of $\Sigma=[\sigma_{ij}]$ corresponding to the edge variables. Note that the ideal $S^{-1}I_G$ has one generator
$g_{ij}$ 
for every non-edge variable $\sigma_{ij}$. This generator corresponds to the constraint $i\dperp j | K$ where $i<j$ and
$K \subset\{k: k<j\}$ are parents of $j$ and we have
$$ g_{ij} = \sigma_{ij} |\sigma_{KK}| - h_{ij}\,. $$
Here $|\sigma_{KK}|$ and $h_{ij}$ are polynomials in $\{\sigma_{a,b}, a\le b<j\}$. Now introduce a lexicographic ordering
on pairs $(i,j)$\footnote{For instance, take $(i,j)<(i',j')$ if $i<i'$ or $i=i'$ and $j<j'$ in the topological sort.} and among all non-edge variables entering into polynomial on the
right consider the maximal one -- denote it $\sigma_{i',j'}$. This variable has its corresponding generator:
$$ g_{i'j'} = \sigma_{i'j'} |\sigma_{K'K'}| - h_{i'j'}. $$
Thus multiplying by a suitable power of $|\sigma_{K' K'}|$ we can write
$$ |\sigma_{K' K'}|^r g_{ij} = |\sigma_{K' K'}|^r (\sigma_{ij} |\sigma_{KK}| - h_{ij}) $$
and now every occurrence of $|\sigma_{K' K'}| \sigma_{i' j'}$ we replace with $g_{i'j'}+h_{i'j'}$. In the end we obtain
$$ |\sigma_{K' K'}|^r g_{ij} = \sigma_{ij}  u'_{ij} - h'_{ij}\,,$$
where the expression on the right no longer contains $\sigma_{i',j'}$ or any larger (w.r.t. ordering of pairs) variable.
Thus, repeating similar steps in the end we obtain expression:
$$ \gamma_{ij} g_{ij} = \sigma_{ij} \tilde u_{ij} - \tilde h_{ij} $$
where we have:
\begin{enumerate}
	\item $\gamma_{ij}=\gamma_{ij}(\sigma)$ is a polynomial with $\ZZ$-coefficients.
	\item $\gamma_{ij}$ and $\tilde u_{ij}$ are units in $S^{-1}\CC[\Sigma]$ (equivalently, they divide $\theta_0^m$ for
	some large enough $m$).
	\item $\tilde u_{ij}$ and $\tilde h_{ij}$ are both polynomials with $\ZZ$-coefficients in \textit{edge} variables $\sigma_{a,b}$ with
	$a\le b < j$.
\end{enumerate}

Consequently, since $\gamma_{ij}$'s are units we have
\begin{equation}\label{eq:sgen}
		S^{-1}I_G = (\sigma_{ij} \tilde u_{ij} - \tilde h_{ij}, (i,j)\mbox{-- non-edge}) 
\end{equation}
Note that on one hand $\mpf_G = I(X_G)$ contains $I_G$. On the other hand, by Proposition~\ref{prop:irred}
and~\eqref{eq:sgen} any minimal prime above $S^{-1}I_G$ has codimension equal to $S^{-1} \mpf_G$. Thus, if we show that
$S^{-1} I_G$ is prime we must have
$$ S^{-1} I_G = S^{-1}\mpf_G $$
and after intersecting with~$\CC[\Sigma]$ conclusion~\eqref{eq:mpfg} and the rest of the theorem follow. 

To that end let $g = g(\sigma_{\edge})$ be the product of all $\tilde{u}_{ij}$. Let $\mathring{\sigma}$ be the identity matrix 
\[
\mathring{\sigma}_{ij}=\left\{\begin{array}{cc} 1& i=j\\0 &i\neq j \end{array}\right.
\]
and let
\[
\varphi^*:\CC[\Sigma]\rightarrow \CC[\Sigma_{\edge}][1/g]
\]
be the ring map associated to the rational map $\Sigma_{\edge}\stackrel{\varphi}{\to} \Sigma$ given by
\begin{equation}
\varphi:\sigma_{\edge}\mapsto \left(\sigma_{\edge}, {\tilde h_{ij}(\sigma_{\edge})\over \tilde u_{ij}(\sigma_{\edge})}, (i,j)\mbox{--
non-edge}\right)\label{eq:ltt}.
\end{equation}
First note that $\mathring{\sigma}$ is a $\CC$-point of $\CC[\Sigma_{\edge}][1/g]$ since
\[
\theta_0(\mathring{\sigma})=1, \,\,g|\theta_0^m\implies g(\mathring{\sigma})\neq 0.
\]
Thus we can form a commutative diagram
\begin{center}
\tikzset{node distance=1.5cm, auto}
\begin{tikzpicture}[auto]
  \node (A0){$\CC[\Sigma]$};
  \node (A1)[right of=A0,node distance=2cm]{$\CC[\Sigma_{\edge}][1/g]$};
  \node (B0)[below of=A0,node distance=0.5cm,right of=A0,node distance=1cm] {$\CC$};
  \path [line] (A0) --node[above]{$\varphi^*$} (A1);
  \path [line] (A0) --node[below,left]{$\mathrm{ev}_{\mathring{\sigma}}$} (B0);
  \path [line] (A1) --node[below,right]{$\mathrm{ev}_{\mathring{\sigma}_{\edge}}$} (B0);
\end{tikzpicture}
\end{center}
Chasing $\theta_0$ from $\CC[\Sigma]$ to $\CC$ in two different ways gives $\theta_0\notin\ker \varphi^*$, and thus $\varphi^*(\theta_0)=h/g^m$ for $h\neq 0$. Localizing $\CC[\Sigma_{\edge}][1/g]$ at $h$ gives a diagram
\begin{center}
\tikzset{node distance=1.5cm, auto}
\begin{tikzpicture}[auto]
  \node (Am1){$\CC[\Sigma]$};
  \node (A0)[right of=Am1,node distance=2cm]{$\CC[\Sigma_{\edge}][1/g]$};
  \node (A1)[right of=A0,node distance=3cm] {$\CC[\Sigma_{\edge}][1/gh]$};
  \path [line] (Am1) --node[above]{$\varphi^*$} (A0);
  \path [line] (A0) --node[above]{$\varphi^*_h$} (A1);
  \path [dashed] (Am1) edge [out=45,in=135,->]node[below]{$\varphi^*_{gh}$} (A1);  
\end{tikzpicture}
\end{center}

Note that $\varphi^*_{gh}$ sends $\theta_0$ to a unit. Hence, by the universal property  of localization (see~\cite[Proposition 3.1]{ati69}), it extends to a map
\begin{align}
	\varphi_{gh}^{*,e}:S^{-1}\CC[\Sigma] \to \CC[\Sigma_{\edge}][1/gh]\label{eq:morp}\\
	\sigma_{ij}\mapsto \varphi^*(\sigma_{ij}),1/\theta_0\rightarrow g^m/h
\end{align}
that is onto and has $S^{-1}I_G$ as kernel. To verify the latter claim, take $s$ in the kernel of (\ref{eq:morp}) and write it as $s=p_\edge+\sum_{ij}q_{ij}.(\sigma_{ij}- \tilde h_{ij}/\tilde u_{ij})$ where $q_{ij}\in S^{-1}\CC[\Sigma]$ with $(i,j)\mbox{--non-edge}$, and $p_\edge$ is a polynomial in $\sigma_{\edge}$. This can be done by virtue of the binomial theorem:
\[
\sigma_{ij}^n=(\frac{\tilde{h}_{ij}}{\tilde{u}_{ij}}-\frac{\tilde{h}_{ij}}{\tilde{u}_{ij}}+\sigma_{ij})^n=(\frac{\tilde{h}_{ij}}{\tilde{u}_{ij}})^n+n(\frac{\tilde{h}_{ij}}{\tilde{u}_{ij}})^{n-1}(\sigma_{ij}-\frac{\tilde{h}_{ij}}{\tilde{u}_{ij}})+\cdots
\]
Then $\varphi^*(s)=p_\edge$, and thus $\varphi^*(s)=0$ gives $p_\edge=0$, that is $s\in S^{-1}I_G$. The reverse inclusion is obvious. 

This establishes isomorphism of rings
$$ S^{-1} \CC[\Sigma]/S^{-1}I_G = \CC[\Sigma_{\edge}][1/gh], $$
which implies that $S^{-1} I_G$ is prime, and that each local ring of $X_G \cap \Sigma^{\mdot}$ is regular (since
all local rings of  $\CC[\Sigma_{\edge}][1/gh]$ are regular).

Geometrically our proof corresponds to constructing a birational isomorphism:
\begin{center}
\tikzset{node distance=1.5cm, auto}
\begin{tikzpicture}[auto]
  \node (Am1){$\Sigma_{\edge}\supset$};
  \node (A0)[right of=Am1,node distance=0.8cm]{$\mathcal{U}$};
  \node (A1)[right of=A0,node distance=2cm] {$V(I_G)\cap\Sigma^{\mdot{}}$};
  \path [line] (A0) --node[above]{$\varphi{|_\mathcal{U}}$} (A1);
  \path [dashed] (A1) edge [out=-135,in=-45,->]node[below]{$\pi$} (A0);  
\end{tikzpicture}
\end{center}
where $\mathcal{U}=D(gh)$ is a distinguished open, $\varphi|_\mathcal{U}$ is obtained by restriction of the map given in~\eqref{eq:ltt}($\varphi$ is regular on $\mathcal{U}$), and $\pi$ is the projection from
$\Sigma$ to $\Sigma_{\edge}$.
\fi
\end{proof}

\ifcomplete

Alternatively, we can choose to prove our results over $\CC[\hat{\Sigma}]$. Recall that $\hat X_G$ is defined as $[\mloc(G) \cap
  \hat{\Sigma}^{++}]_Z$. 

\begin{theorem} Let $G$ be a DAG and let $\hat{\theta}_0=\prod_{A\subset [n]}(|{\hat{\Sigma}}_{AA}|)$. With the notation of Theorem~\ref{thm:thm1},
\noindent
\begin{enumerate}[(a)]
\item  $\hat {X}_G$ is an affine rational (irreducible) variety of dimension $|E|$ and $\hat X_G = X_G \cap \hat \Sigma$.

\item We have $$ V(\hat{I}_G) = \hat{X}_G \cup \hat{B}_G, $$
where $\hat B_G \subset\{\hat{\theta}_0=0 \}$.
\item The prime ideal $\hat p_G \eqdef I(\hat X_G)$ is obtained by saturating $\hat{I}_G$
\begin{equation}\label{eq:hmpfg}
		\hat{\mpf}_G = \hat{S}^{-1} \hat{I}_G \cap \CC[\hat{\Sigma}] 
\end{equation}
at the multiplicatively closed set $\hat{S}=\{\hat{\theta}_0^n, n=1,\ldots\}$.
\item $\hat{X}_G$ is smooth inside $\hat{\Sigma}^{\mdot}$.
\end{enumerate}
\label{thm:xhat}
\end{theorem}
\begin{proof}
We first claim that $X_G\cap \hat{\Sigma}$ is an affine rational irreducible variety that is smooth inside $\hat{\Sigma}^{\mdot}$. Consider the sequence
\begin{equation}
S^{-1}\CC[\Sigma]\stackrel{\varphi^{*,e}_{gh}}{\to} \CC[\Sigma_{\edge}][1/gh]\stackrel{\mathrm{ev}}{\to} \CC[\hat{\Sigma}_{\edge}][1/\hat{g}\hat{h}],
\label{eq:hatseq}
\end{equation}
where $\varphi^{*,e}_{gh}$ is the extended map in (\ref{eq:morp}) and $\hat{g}=g|_{\hat{\Sigma}},\hat{h}=h|_{\hat{\Sigma}}$, and $\mathrm{ev}$ is the  evaluation map $\sigma_{ii}=1$. Observe that  $$X_G\cap \hat{\Sigma}=V((I_G:\theta_0^m))\cap V(\ker (\mathrm{ev})^c)=V((I_G:\theta_0^m)+\ker (\mathrm{ev})^c)=V(\hat{I}_G:\hat{\theta}_0^m).$$
Also note that $$(\hat{I}_G:\hat{\theta}_0^m)=\hat{S}^{-1}\hat{I}_G\cap \CC[\hat{\Sigma}].$$ Let us show that $\hat{S}^{-1}\hat{I}_G$ is prime of codimension $|E|$.   
Note that $\varphi_{gh}^{*,e}$ is identity on the generators of $\ker(\mathrm{ev})^c$, i.e., $\varphi_{gh}^{*,e}(\ker(\mathrm{ev})^c)=\ker(\mathrm{ev})$. 
This implies that $\ker(\mathrm{ev} \varphi^{*,e}_{gh})\supset \ker(\varphi^{*,e}_{gh})+\ker (\mathrm{ev})^c$. The reverse inclusion $\ker (\mathrm{ev} \varphi^{*,e}_{gh})\subset \ker(\varphi^{*,e}_{gh})+\ker( \mathrm{ev})^c$ is obvious. 
Since both maps in (\ref{eq:hatseq}) are onto, we have an isomorphism of integral domains
\begin{equation}
S^{-1}\CC[\Sigma]/(S^{-1}I_G+\ker({\mathrm{ev}})^c)= \CC[\hat{{\Sigma}}]/\hat{S}^{-1}\hat{I}_G= \CC[\hat{\Sigma}_{\edge}][1/\hat{g}\hat{h}],
\label{eq:hatiso}
\end{equation}
where we used the fact that $\ker(\varphi^{*,e}_{gh})=S^{-1}I_G$, shown in the proof of Theorem~\ref{thm:thm1}. This proves the primality of $\hat{S}^{-1}\hat{I}_G$, hence, $X_G\cap \hat{\Sigma}$ is irreducible of dimension $|E|$. This further implies that the map (\ref{eq:hatiso}) is induced by the restriction of the rational map $\varphi|_{\hat{\Sigma}}$. In particular, this restriction is an isomorphism inside $\hat{\Sigma}^{\mdot}$, proving the above claim. 

Now the proof reduces to shwoing that $\hat{X}_G=X_G\cap \hat{\Sigma}$. We know that $X_G\cap \hat{\Sigma}$ is irreducible by the above discussion, and that it contains $\hat{X}_G$ by Proposition \ref{pro_dense}. Thus the theorem follows if we show that $\hat{X}_G$ has dimension $|E|$.  Recall from Proposition \ref{pro_minors} and construction of $\varphi$ that $\mloc(G)\cap \hat{\Sigma}^{++}$ can be obtained as the intersection of the image of $\varphi$ with the real subset of $\hat{\Sigma}^{++}$, i.e., it has the structure of a real differentiable manifold of dimension $|E|$. Together with Proposition \ref{pro_dense} and  \cite[Proposition 2.8.14]{boc98}, this implies that
\[
\dm (\mathbb{R}[\hat{\Sigma}]/I_\mathbb{R}(\hat{X}_G))=|E|,
\]
where $I_\mathbb{R}(\hat{X}_G)$ is the ideal of polynomials with real coefficients that vanish on $\hat{X}_G$.
Since the set $\mloc(G)\cap \hat{\Sigma}$ is stable under conjugation, we have $I(\hat{X}_G)=I_\mathbb{R}(\hat{X}_G)^e$. Then the going up theorem implies that $\hat{X}_G$ has dimension $|E|$ as desired.
\end{proof}
\fi

The next example shows how Theorem~\ref{cor_conj} can be used to construct $\mathfrak{p}_G$ from $I_G$:
 
\begin{example}
Consider the DAG
\begin{center}
\tikzset{node distance=1.5cm, auto}
\begin{tikzpicture}[auto]
  \node (A0){$$};
  \node (A1)[right of=A0,node distance=1.5cm] {$$};
  \node (A2) [right of=A1,node distance=1.5cm]{$4$};
  \node (A3) [below of=A0,node distance=0.5cm]{$G:1$};
  \node (A4) [right of=A3,node distance=1.5cm]{$2$};
  \node (A5) [below of=A2,node distance=1cm]{$3$};
  \path [line] (A3) --node[above]{$$} (A4);  
  \path [line] (A4) --node[above]{$$} (A5);  
  \path [line] (A4) --node[above]{$$} (A2);  
  \path [line] (A5) --node[above]{$$} (A2);  
\end{tikzpicture}
\end{center}
The ideal of imposed relations is generated by relations $1\ci 3|2$ and $4\ci 1|(2,3)$:
\[
I_G=\langle |\sigma_{12,23}|,|\sigma_{123,423}|\rangle.
\]
It has primary components
\begin{align*}
I_{G,1}&=\langle \sigma_{12}\sigma_{23} - \sigma_{13}\sigma_{22},
        \sigma_{12}\sigma_{24} - \sigma_{14}\sigma_{22},
        \sigma_{13}\sigma_{24} - \sigma_{14}\sigma_{23}\rangle
\end{align*}
and
\[
I_{G,2}=\langle \sigma_{12}\sigma_{33} - \sigma_{13}\sigma_{23},
        \sigma_{12}\sigma_{33} - \sigma_{13}\sigma_{23},
        \sigma_{22}\sigma_{33} - \sigma_{23}^2\rangle.
\]
It can be seen that only $I_{G,1}$ intersects $\Sigma^{\mdot}$\,. We thus have $\mathfrak{p}_G=I_{G,1}$. Furthermore, $I_{G,1}$ is the unique ideal contained in the maximal ideal at the identity of $\Sigma$, and is also equal to the saturation of $I_{G}$ at $f=\sigma_{22}(\sigma_{22}\sigma_{33}-\sigma_{23}^2)$. 
\ifcomplete
The ideal of implied relations is generated by relations $1\ci 3|2,1\ci 4|2,1\ci 4|(2,3)$, and $1\ci 3|(2,4)$:
\[
J_G=\langle |\sigma_{12,23}|,|\sigma_{12,42}|,|\sigma_{123,423}|,|\sigma_{124,324}|\rangle.
\]
It has primary components
\[
J_{G,1}=\langle \sigma_{13}\sigma_{22} - \sigma_{12}\sigma_{23},\sigma_{14}\sigma_{22} - \sigma_{12}\sigma_{24},\sigma_{14}\sigma_{23} - \sigma_{13}\sigma_{24} \rangle
\]
and
\[
J_{G,2}=\langle \sigma_{12}, \sigma_{22},\sigma_{24},\sigma_{23}^2\rangle.
\]
Again one can check that $J_{G,1}=I_{G,1}=S^{-1}J_G\cap \CC[\Sigma]$ is the unique component that is contained in the maximal ideal at the identity. Finally, we can see that $$S^{-1}\hat{I}_G\cap \CC[\hat{\Sigma}]=\langle\hat{\sigma}_{12}\hat{\sigma}_{23} - \hat{\sigma}_{13}, \hat{\sigma}_{12}\hat{\sigma}_{24} - \hat{\sigma}_{14},\hat{\sigma}_{13}\hat{\sigma}_{24} -\hat{\sigma}_{14}\hat{\sigma}_{23}\rangle$$ is the unique irreducible component of $V(\hat{I}_G)$ that contains the origin of $\hat{\Sigma}$. \\
\fi
\label{ex:thm1}
\end{example}

\ifcomplete
Let us point out that $X_G$ need not be smooth outside of $\Sigma^{\mdot}$ as the next example shows:
\begin{example} Let $G$ be the Markov chain $1\rightarrow 2\rightarrow 3$. Then $X_G=V(\sigma_{22}\sigma_{13}-\sigma_{12}\sigma_{23})$ is a cone and has a singularity at the origin. 
\end{example}

\else

\fi

\ifcomplete

\fi
\subsection{Algebraic representation} 
\label{sec:algrep}
Here, we put together the results of the previous sections to give an algebraic
criteria for testing isomorphism of graphical models. We start by a result on equivalence of DAGs:

\begin{proposition}
Let $G,G'$ be DAGs. Then $G$ is equal to $G'$ if and only if $X_G=X_{G'}$, or equivalently, if and only if $\hat{X}_G=\hat{X}_{G'}$.
\label{pro_dageq}
\end{proposition}
\begin{proof}
\ifcomplete
If $G$ is equal to $G'$, then $\mloc(G)=\mloc(G')$ by Proposition~\ref{pro_dagiso}, and thus $X_G=X_{G'}$. Conversely, if $X_G=X_{G'}$, then $V(I_G)\cap\Sigma^{\mdot}= V(I_{G'})\cap\Sigma^{\mdot}$ by Theorem~\ref{cor_conj}a. This implies that $V(I_G)\cap \Sigma^{++}= V(I_{G'})\cap\Sigma^{++}$ and by (\ref{eq:vcmloc}) that $\mloc(G)\cap\Sigma^{++}=\mloc(G')\cap\Sigma^{++}$. Then by Proposition ~\ref{pro_dagiso} we have $G=G'$. The last part of the assertion follows by a similar reasoning.
\fi
\end{proof}

In what follows, $\Pi=\{\pi_s\}_{s\in S_n}$ is the permutation
group with induced action on $\CC[{\Sigma}]$:  $\pi_s(f(\sigma_{ij}))=f(\sigma_{s(i)s(j))})$ where $s\in S_n$ is a
permutation of indices. The invariant subring $\{f\in {\CC}[{\Sigma}]\,|\, f\circ \pi_s=f\,\, \forall s\}$ is denoted by
${\CC}[{\Sigma}]^\Pi$. We can now state our main result:

 \begin{theorem} Let $G,G'$ be DAGs and $S$ be as in Theorem~\ref{cor_conj}. Then $G\sim G'$ if and only if
	\begin{align}\label{eq:tpr}
S^{-1}I_G\cap {\CC}[\Sigma]^\Pi=S^{-1}I_{G'}\cap {\CC}[{\Sigma}]^\Pi.
\end{align}
\label{thm:iso}
\end{theorem}
\begin{proof}\ifcomplete By Proposition~\ref{pro_dageq}, DAGs $G$ and $G'$ are isomorphic iff there
	exists $\pi \in \Pi$ such that
	$$ X_{\pi(G)}= X_{G'}. $$
	By Theorem~\ref{cor_conj}b this in turn is equivalent to 
\begin{equation}\label{eq:tp1}
	\pi(X_G) = X_{G'}. 
\end{equation}
To see this note that 
	$$X_{\pi(G)}=V( I_{\pi(G)}:\theta_0^m)=V(\pi(I_G:\theta_0^m)))=\pi(X_G).$$
Furthermore, since $S$ in Theorem~\ref{cor_conj} is $\Pi$-invariant,~\eqref{eq:tp1} is equivalent to
\begin{equation}\label{eq:tp3}
		\pi(S^{-1} I_G)\cap \CC[\Sigma] = S^{-1} I_{G'}\cap \CC[\Sigma]\,. 
\end{equation}
Finally, because of primality of ideals in~\eqref{eq:tp3}, the existence of $\pi$ satisfying~\eqref{eq:tp3} is equivalent to~\eqref{eq:tpr} by~\cite[Exercise 5.13]{ati69}.
	\fi
\end{proof}

\ifcomplete
As before, we can work over $\mathbb{C}[\hat{\Sigma}]$.

\begin{theorem}: Let $G,G'$ be DAGs and $\hat{S}$ be as in Theorem~\ref{thm:xhat}. Then $G$ and $G'$ are isomorphic if and only if
	\begin{align}\label{eq:tpr}
\hat{S}^{-1}\hat{I}_G\cap {\CC}[\hat{\Sigma}]^\Pi=\hat{S}^{-1}\hat{I}_{G'}\cap {\CC}[\hat{\Sigma}]^\Pi.
\end{align}
\label{thm:hiso}
\end{theorem}
\begin{proof}
Recall that $X_{\pi(G)}=\pi(X_G)$. Thus by Proposition~\ref{pro_dageq} we have
\[
\pi(X_G)=X_{G'} \iff \pi(\hat{X}_G)=\hat{X}_{G'}.
\]
The rest of the proof is analogous to that of Theorem~\ref{thm:iso} and is omitted.
\end{proof}
We remark that $\hat{S}^{-1}\hat{I}_G$ can be replaced with $\hat{J}_G$ in the above. The presentation of the results in this form is a matter of convenience. Indeed, there is a simple way to generate $\hat{I}_G$: 

\begin{itemize}
\item Traverse the graph in the order of the topological sort and set $\hat{I}_i:= \sum_{{j<i}, j\notin K}\langle |\hat{\sigma}_{iK,jK}|\rangle$ with $K:=\mathbf{pa}(i)$ for all $i$.
\item Output $\hat{I}_{G}=\sum_i \hat{I}_i$.
\end{itemize}
However, extracting the implied relations of a DAG requires more work.

We also point out that the extra components of $V(\hat{I}_G)$ (or $V(I_G)$) are not invariant across the isomorphism class of $G$:
\begin{example}
The DAG
\begin{center}
\tikzset{node distance=1.5cm, auto}
\begin{tikzpicture}[auto]
  \node (A0){$$};
  \node (A1)[right of=A0,node distance=0cm] {$$};
  \node (A2) [above of=A1,node distance=0.75cm]{$2$};
  \node (A3) [below of=A0,node distance=0.5cm]{$1$};
  \node (A4) [right of=A3,node distance=1.5cm]{$$};
  \node (A6) [left of=A3,node distance=0.5cm,above of=A3,node distance=0.7cm]{${G}':$};
  \node (A7)[right of=A6,node distance=2cm] {$3$};
  \node (A5) [right of=A7,node distance=1.25cm]{$4$};
  \path [line] (A3) --node[above]{$$} (A2);  
  \path [line] (A3) --node[above]{$$} (A7);  
  \path [line] (A7) --node[above]{$$} (A5);  
  \path [line] (A2) --node[above]{$$} (A7);  
\end{tikzpicture}
\end{center}
is isomorphic to the DAG $G$ in Example \ref{ex:thm1}, but $\hat{I}_{G'}$ is a prime ideal. Note however that $J_{G'}=J_{\pi(G)}$ where $\pi$ is the permutation $(14)(23)$. Thus 
it is necessary to compute the saturation ideal in (\ref{eq:tpr}). We shall see, however, in section \ref{sec:rand} that one can avoid computing the saturation ideal, and more importantly, the subring intersection by a probabilistic procedure. 
\label{ex:inv}
\end{example}
\fi
\ifcomplete
\subsection{Changing the ground ring}

We briefly note that none of our algebraic arguments depend on the choice of $\CC$ as ground ring. In particular,
instead of $X_G$ we could have considered 
$$ X'_G = \mathrm{Spec}(\ZZ[\Sigma]/\mq_G) $$
where $\mq_G = S^{-1} I_G \cap \ZZ[\Sigma]$. Then we still have that $X'_G$ is an integral, rational 
scheme over $\ZZ$,
smooth at every point in $D(\theta_0) = \{\theta_0\neq 0\}$. Two graphs $G\sim G'$ are isomorphic if and only if 
$$ S^{-1} I_G \cap \ZZ[\Sigma]^\Pi = S^{-1} I_{G'} \cap \ZZ[\Sigma]^\Pi. $$
Furthermore, we can identify $X_G$ as the base-change of $X'_G$ to $\mathbb{C}$: $X_G=X'_G\times_{\mathrm{Spec}(\mathbb{Z})} \mathrm{Spec}(\mathbb{C})$. 
\ifcomplete
To verify this, one merely needs to check that $\mq_G^e=\mpf_G$. Indeed $I_G$ has a \gb basis (w.r.t any term order) consisting of generators with $\mathbb{Z}$-coefficients (run Buchberger's algorithm and clear denominators at the end). One can thus see, by variable elimination, that $\mpf_G=S^{-1}I_G\cap \CC[\Sigma]$ can be generated by polynomials in $\ZZ[\Sigma]$. 
\fi
Other convenient choices of rings are $\mathbb{Q}$ and $\mathbb{R}$.  
\fi
\subsection{Randomized algorithm}
\label{sec:rand}
\ifcomplete
\begin{algorithm}[t]
\caption{$\isodag_{\textup{m}}$}
\label{ISODAG}
\begin{algorithmic}[1]
\STATE {\bf function} {$\isodag_\mathrm{m}$}($G$, $G'$)
\STATE Sort $G$ and $G'$ topologically
\STATE Initialize $\mathrm{ISO}\gets \TRUE,r\gets 1$
\WHILE{$\mathrm{ISO}\,\,\AND \,\,r\le m$}
\STATE Sample $z_G^r,z_{G'}^r$ respectively from $\hat{\varphi}_*P_G$,$\hat{\varphi}_*P_{G'}$ as follows:\\
\begin{enumerate}[(i)]
\item Sample edge variables $\hat{\sigma}^r_{\edge}$ of $G$ from $P_G$
\FOR{$i:=2$ to $n$}
\item[] \quad Solve the (linear) toposorted imposed relations
\[		
|\hat{\sigma}^r_{iK,jK}|=0, \quad K:=\p(i)
\]
\quad for each non-edge variable $\hat{\sigma}^r_{ij}$, $j<i$
\item[] \ENDFOR

\item $z_G^r\gets (\hat{\sigma}_{\edge}^r,\hat{\sigma}_{\mathrm{non-edge}}^r)$
\item Repeat for $G'$
\end{enumerate}
\IF{\,\,$\Pi(z^r_G)\cap V(\hat{I}_{G'})=\emptyset\,\OR\, \Pi(z^r_{G'})\cap V(\hat{I}_G)=\emptyset$} 
	\STATE $\mathrm{ISO}\gets \FALSE$
\ENDIF
\STATE $r\gets r+1$
\ENDWHILE
\RETURN ISO\\
\end{algorithmic}
\end{algorithm}

\fi

\ifcomplete Theorem~\ref{thm:hiso} shows that testing isomorphism amounts to subring intersection. Computing this intersection is difficult since there is no easy description available for generating invariants of $\CC[\hat{\Sigma}]^\Pi$. Another computational difficulty is that of computing the saturation ideal. This operation does not scale well with the number of nodes in the model. Here we give a randomized algorithm that avoids computing both the intersection and saturation ideals. 

Let $\varphi$ be the rational map in (\ref{eq:ltt}) and denote by $\hat{\varphi}$ its restriction $\varphi|_{\hat{\Sigma}}$. Let $\mathcal{Y}:=\mathbb{F}_q^{|E|}$ and define
$$\mathcal{U}:=\{y\in \mathcal{Y}: \hat{g}(y)\neq0,\hat{h}(y)\neq 0\},$$ 
where  $\hat{g},\hat{h}$ are as in  (\ref{eq:hatseq}). Note that
\[
|\mathcal{U}|\ge (q-d) q^{|E|-1}
\]
with $d:=\deg(\hat{g})+\deg(\hat{h})$.

\fi 
Now construct a random matrix with uniform distribution $P_G$ on the finite set $\mathcal{U}$. \ifcomplete This can be realized, for instance, by transforming a uniformly distributed matrix on $\mathcal{Y}$ through a kernel $P_{Y|X}:\mathcal{Y}\to\mathcal{Y}\cup \{\emptyset\}$ such that  $P_{Y|X}(x|x)=1$ if $x\in \mathcal{U}$ and $P_{Y|X}(\emptyset |x)=1$ otherwise.\fi Let $\hat{\varphi}_{*}P_G$ be the push-forward of $P_G$ under $\hat{\varphi}$ and $Z_G^i$'s be independent random variables with common distribution $\hat{\varphi}_*P_G$. 

Given DAGs $G$, $G$', the algorithm $\isodag_\mathrm{m}$ described above constructs $m$ realizations $z^i_G, z^i_{G'}$ from $Z^i_G,Z^i_{G'}$. It then declares $G$ and $G'$ to be isomorphic if and only if for each $i\le m$, there is some permutation $\pi$ such that both $\pi(z^i_G)\in V(\hat{I}_{G'})$ and $z^i_{G'}\in V(\pi(\hat{I}_G))$ hold. \ifcomplete The latter conditions amount, respectively, to checking $f'_j(\pi(z^i_G))=0$ and $\pi(f_k)(z^i_{G'})=0$ for all generators $f'_j$ of $\hat{I}_{G'}$ and $f_k$ of $\hat{I}_G$. The next theorem shows that the probability of failure of the algorithm can be made arbitrarily small:\fi



\begin{theorem}
Let $G$ be a DAG on $n$ nodes and $E$ be its set edges. Let $Z^i_G$ be as in above and set $d:=\deg(\hat{g})+\deg(\hat{h})$. If $G\sim G'$, then
\[
\mathbb{P}[\isodag_\mathrm{m}(G,G')=\mathrm{yes}]=1.
\]
If $G\not\sim G'$, then
\[
\mathbb{P}[\isodag_\mathrm{m}(G,G')=\mathrm{yes}]\le  (n!\frac{n+2d-1}{q-d})^m.
\]
\label{thm:rand}
\end{theorem}
\begin{proof}
\ifcomplete
By Theorem~\ref{thm:hiso}, the algorithm outputs yes with probability 1 if $G$ is isomorphic to $G'$. Now suppose $\mloc(G)\not\subset \mloc(G')$. We need to upper bound the probability that a realization of $Z^i_G$ lands on $\hat{X}_{G'}$. Note that $Z^i_G$ takes values on $\hat{X}_{G}\cap \hat{\Sigma}^{\mdot}$. It follows from  Proposition \ref{pro_dageq} that a point on $\hat{X}_{G'}\cap \hat{\Sigma}^{\mdot}$ satisfies at least one relation of the form $|\hat{\sigma}_{iK,jK}|$ where $i\ci j|K \not\in \mathcal{M}^{topo}_G$. We can pullback the intersection $\hat{X}_G\cap \hat{\Sigma}^{\mdot}\cap \{|\hat{\sigma}_{iK,jK}|=0\}$ to $\mathcal{U}$ via the embedding $\hat{\varphi}$. This corresponds to replacing the non-edge variables (w.r.t $G$) of $|\sigma_{iK,jK}|$ with rational functions of edge variables as in (\ref{eq:ltt}). Let $f_{ijK}$ be the resulting rational function. 
Then the number of points on the intersection is upper bounded by the number of solutions in $\mathcal{Y}$ of $\{f_{ijK}=0\}$. Note that the minor has total degree at most $n-1$ and the largest power of any monomial appearing in the expansion of such minor is at most $2$. Thus clearing the denominators in $f_{ijK}$ gives a non-zero polynomial of degree at most $(n+2d-1)$ in at most $|E|$ variables. Such a polynomial has at most $(n+2d-1)q^{|E|-1}$ roots in $\mathcal{Y}$. 
On the other hand, $\mathcal{U}$ has at least $(q-d)q^{|E|-1}$ points. The union bound proves the result. 
\fi
\end{proof}
If $G$ and $G'$ are two DAGs on $n$ nodes and $E$ edges, it is easy to verify that the algorithm can decide equivalence of $G$ and $G'$ in time $O(n^3|E^c|)$. To test if a sampled point lies on $\hat{X}_G$, one needs to solve a sequence of $|E^c|$ linear equations. Each equation involves one missing edge and one minor of size at most $(n-1)\times (n-1)$, which requires $O(n^3)$ operations to compute (note that we can do the arithmetics over rationals). This is comparable, in terms of complexity, with  $O(n^4|E|)$ operations needed in the essential graph method (see \cite{and97}).

\ifcomplete
\section{Directed tree models}
We study the case of directed tree models in this section. The main property is the following:
\begin{proposition}
Let $T$ be a directed tree model. Then $\hat{I}_T$ is prime and $V(I_T)\cap\hat{\Sigma}=[\mloc(T)\cap\hat{\Sigma}]_Z$. Furthermore, $\hat{X}_T$ is smooth everywhere.
\label{pro_prime}
\end{proposition}
This is analogous to primality of $J_T$ shown in \cite{sul08} (see Corollary 2.4 and Theorem 5.8). A direct consequence is that two tree models are isomorphic if and only if $ \hat{I}_T\cap \CC[\hat{\Sigma}]^{\Pi}=\hat{I}_{T'}\cap \CC[\hat{\Sigma}]^\Pi.$
It also follows from the proposition that $\hat{I}_T=\hat{J}_T$. 
\noindent
In II.E, we give a procedure to generate $\hat{I}_T$. To prove primality of $\hat{I}_T$, we introduce a second procedure that uses lower degree generators by modifying the first step:
\begin{itemize}
\item Set ${L}_0=\emptyset$.
\item Traverse the graph in order of the sort and set $K_i:=\mathbf{pa}(i)$ for all $i$. 
\item For all $j<i$, let $K'\in K_i$ be the smallest subset that $d$-separates $j$ and $i$. Set $L_i=L_{i-1}\cup \{f| f= \langle |\sigma_{iK',jK'}|\rangle\}$.
\item Output $\hat{I}'_G:=\sum_{f\in {L}_n} \langle f \rangle$.
\end{itemize}

Let us verify $\hat{I}_T=\hat{I}'_T$ for a tree model $T$. 
Suppose the claim holds on $n-1$ nodes 
and let the first procedure reach node $n$. Set $K=\mathrm{pa}(i)$ and let $k$ be the separator of nodes $i$ and $n$ for some $i<n$. The generated polynomial is 
\[
|\sigma_{ iK,nK}|=\left|\begin{array}{cc} \sigma_{in}& \sigma_{iK}  \\ \sigma_{Kn} & \sigma_{KK} \end{array} \right|.
\]
This can be expanded as
\begin{align*}
|\sigma_{ iK,nK}|&=\left|\begin{array}{cc} \sigma_{in}& \sigma_{ik}  \\ \sigma_{kn} & 1 \end{array} \right|+\sum_{j\in \{K-k\}} \sigma_{ij}f_j+\sum_{j,j'\in K} \sigma_{jj'}f_{jj'},
\end{align*}
which is in the ideal of the second procedure since $i$ is not connected to $K-k$ and the parents of $n$ are not connected either. Conversely, take the minor 
\[
\left|\begin{array}{cc} \sigma_{in}& \sigma_{ik}  \\ \sigma_{kn} & 1 \end{array} \right|
\]
produced by the second procedure and note that \mbox{$i\ci K-k$} and $k'\ci k''$ for distinct $k',k''\in K$. By the inductive hypothesis, the corresponding linear forms are in $I_T$.  Using the above expansion for $|\sigma_{ iK,nK}|$ one can see that this reduced minor is in the ideal of the first procedure as well. Thus the two ideals are equal when $T$ is a tree.  
The first procedure is easy to implement and is what we use to generate the ideals. The second procedure is easy to analyze and has useful properties that we exploit later. For instance, it shows that the ideal of imposed relations of a tree model is generated by quadratic polynomials of type $\sigma_{ij}-\sigma_{ik}\sigma_{kj}$ or linear forms $\sigma_{ij}$. We note that the equality of $\hat{I}_T$ and $\hat{I}'_T$ does not generalize to all DAGs. For instance, in the case of Example~\ref{ex:thm1}, $\hat{I}_T$ is not prime whereas $\hat{I}'_T$ is prime for all DAGs with $n\le 4$ nodes.

We are now ready to prove that the ideal of the imposed relations is prime for tree models:\\
\begin{proof}[Proof of Proposition~\ref{pro_prime}]
The generators of $\hat{I}_T$ are of the form either $\hat{\sigma}_{ik}-\hat{\sigma}_{ij}\hat{\sigma}_{jk}$ or $\hat{\sigma}_{ij}$. Then the rational map (\ref{eq:morp}) in the proof of Theorem~\ref{cor_conj} can be taken to be a polynomial map. In other words, $\hat{S}^{-1}\hat{I}_T\cap\mathbb{C}[\hat{\Sigma}]=\hat{I}_T$. It further follows that $\hat{X}_T$ is smooth.
\end{proof}

One can further show that for a certain lexicographic order the generators of $\hat{I}_G$ form a \gb basis. Set $d_{ij}=\min_{\pi \in \mathcal{D}(i,j|\emptyset)} |\pi| $, that is $d_{ij}$ is the length of the shortest d-path from $i$ to $j$. Note that if $i\ci j|k$ is an imposed relation, then $d_{ij}>1$. Order the variables as follows: $\sigma_{ij}\succ\sigma_{i'j'}$ if $d_{ij}>d_{i'j'}$ or $d_{ij}=d_{i'j'}$ but $(i,j)\succ_{\ZZ^2}(i'j')$, where $\succ_{\ZZ^2}$ is any order on $\ZZ^2$. We also set $\sigma_{ij}\succ 1$ for all variables. \\
Let $\alpha=(\alpha_{ij})$ be a vector. Denote by $\sigma^\alpha$ the monomial $\prod \sigma_{ij}^{\alpha_{ij}}$. Now define the relation $\sigma^\alpha\succ_\dag \sigma^\beta$ if the first non-zero coordinate of $\alpha-\beta$ is positive. Note that $\succ_\dag$ is a lexicographic order. 

The above order has a pleasant property: the leading monomial of quadratic relations $\sigma_{ij}-\sigma_{ik}\sigma_{kj}$ generated by CI relations of $T$ is always the linear form $\sigma_{ij}$. We use this to prove that the quadratic and linear imposed relations used to generate $\hat{I}'_T$ form a \gb basis w.r.t to this order:

\begin{proposition} The generators of $\hat{I}'_T$ form a \gb basis w.r.t the $\succ_\dag$ oder. 
\label{pro_gb}
\end{proposition}
\begin{proof}
Note that $\hat{I}'_G$ is generated by linear forms and quadratic relations of the form $\sigma_{ij}-\sigma_{ik}\sigma_{kj}$. Given two such polynomials, $f=\sigma_{ij}-\sigma_{ik}\sigma_{kj}$ and $f'=\sigma_{i'j'}-\sigma_{i'k'}\sigma_{k'j'}$, we can check that the leading terms are linear, and hence, the resulting $S$-polynomial
\[
S(f,f')=\sigma_{ik}\sigma_{kj}\sigma_{i'j'}-\sigma_{i'k'}\sigma_{k'j'}\sigma_{ij},
\]
reduces to zero w.r.t $\{\sigma_{ij}-\sigma_{ik}\sigma_{kj},\sigma_{i'j'}-\sigma_{i'k'}\sigma_{k'j'} \}$. 
\end{proof}
It was asked in \cite{sul08} (see Conjecture 5.9) if there exists a \gb basis consisting of square free terms of degree one and two for $J_T$. The above proposition shows that this is the case for $\hat{I}_T$ (or $\hat{J}_T$). 

Finally, using the procedure in section \ref{sec:rand}, we list the isomorphism classes of trees on $4,5,$ and $6$ nodes. See Figure~\ref{fig_classes}. Our computations are done in Magma.

\begin{figure}[t]
        \centering
\subfigure[$n=4$]{
        \centering
\includegraphics[height=0.12\textheight]{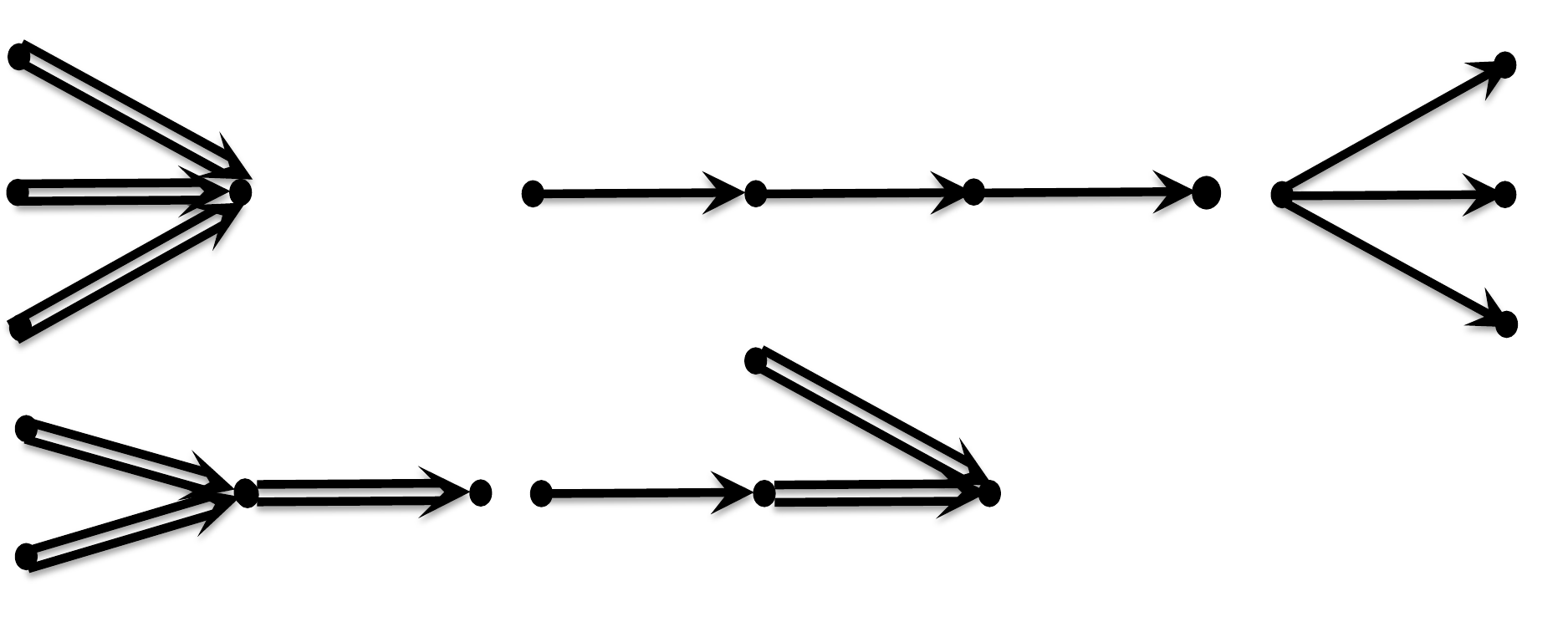} 
                }
\subfigure[$n=5$]{
        \centering
\includegraphics[height=0.28\textheight]{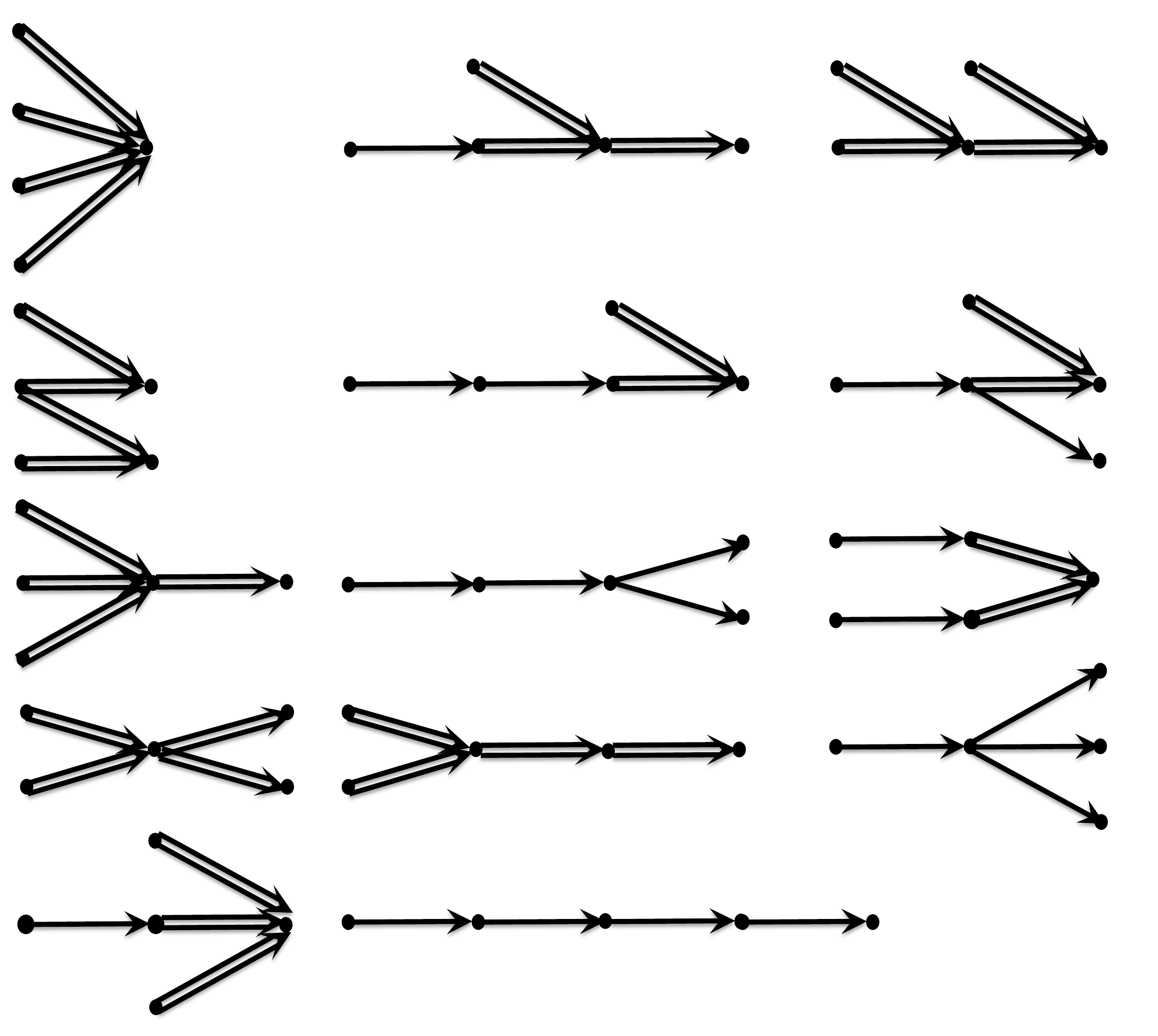} 
                }
                \subfigure[$n=6$]{
\centering
\includegraphics[height=0.45\textheight]{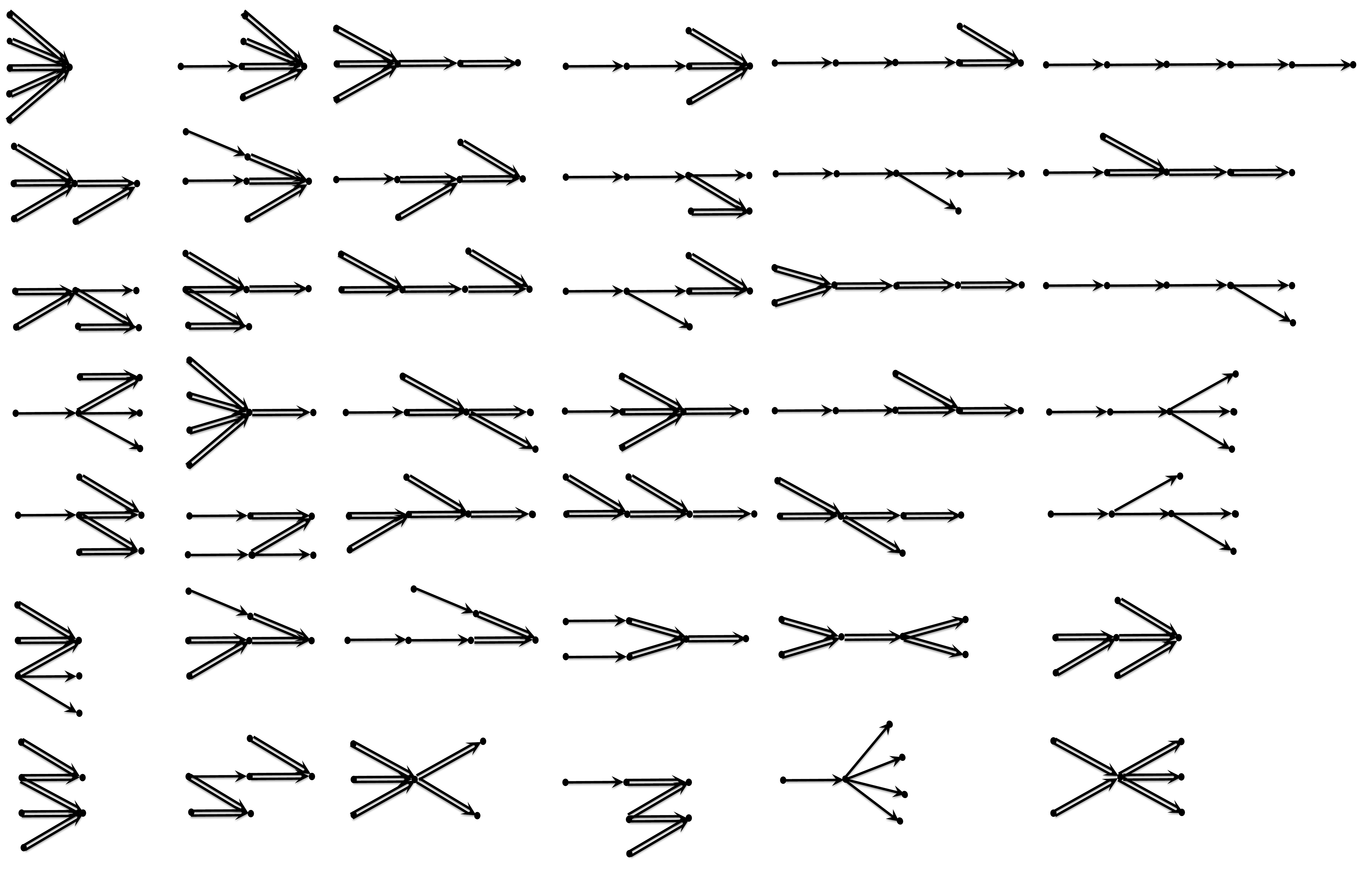} 
}
        \caption{Isomorphism classes of directed tree models on $4,5,$ and $6$ nodes. The double arrows are essential.}                 
                \label{fig_classes}
\end{figure}

\fi


\section{Marginalization}

Exact inference on a DAG $G$ is the problem of extracting marginals of a $G$-compatible distribution. In this section, we establish a connection between the obstructions to embeddings of $\X_G$ and obstructions to efficient inference on $G$. 

\begin{definition} \begin{enumerate} 
\item Given a DAG $G$, we define marginalization of $G$ w.r.t $N$, denoted by $\El_N(G)$, to be the conditional independence model defined by the implied relations in $G$
\[
i\ci j|K
\]
such that $\{i,j,K\}\cap N=\emptyset$. 
\item Let $e=|N|$. Given a DAG $G$, define the extension $G^{e}$ of $G$ w.r.t $N$ to be the DAG on $n+e$ nodes whose implied relations are the same as $G$. 
\end{enumerate}
\end{definition}
Take for instance the DAG
\begin{center}
\tikzset{node distance=1.5cm, auto}
\begin{tikzpicture}[auto]
  \node (A0) {$1$};
  \node (A1) [left of=A0,node distance=0.5cm]{$G:$};
  \node (A3l) [right of=A0,node distance=1.5cm]{$$};
  \node (A3ul) [above of=A3l,node distance=1cm]{$2$};
  \node (A3dl) [below of=A3l,node distance=1cm]{$3$};
  \node (A4) [right of=A0,node distance=3cm]{$4$};
  \node (A5p) [right of=A4,node distance=1.5cm]{$$};
  \path [line] (A0) --node[above]{$$} (A3ul);  
  \path [line] (A3ul) --node[above]{$$} (A4);  
  \path [line] (A3dl) --node[above]{$$} (A4);  
  \path [line] (A0) --node[above]{$$} (A3dl);  
\end{tikzpicture}
\end{center}
To extract the marginal on nodes 1 and 3, one needs to eliminate nodes 2 and 4. A well known problem is that the order of elimination matters. For instance, eliminating node 4 leads to a ``sparse'' graph 

\begin{center}
\tikzset{node distance=1.5cm, auto}
\begin{tikzpicture}[auto]
  \node (A0){$$};
  \node (A1)[right of=A0,node distance=0.5cm]{$2$};
  \node (A2)[right of=A1,node distance=1.5cm] {$$};
  \node (A2u)[above of=A2,node distance=0.75cm] {$1$};
  \node (A3) [right of=A2,node distance=1.5cm]{$3$};
  \path [line] (A2u) --node[above]{$$} (A1);  
  \path [line] (A2u) --node[above]{$$} (A3);  
\end{tikzpicture}
\end{center}
whereas eliminating node 2 gives a ``dense'' graph

\begin{center}
\tikzset{node distance=1.5cm, auto}
\begin{tikzpicture}[auto]
  \node (A0){$$};
  \node (A1)[right of=A0,node distance=0.5cm]{$3$};
  \node (A2)[right of=A1,node distance=1.5cm] {$$};
  \node (A2u)[above of=A2,node distance=0.75cm] {$1$};
  \node (A3) [right of=A2,node distance=1.5cm]{$4$};
  \path [line] (A2u) --node[above]{$$} (A1);  
  \path [line] (A2u) --node[above]{$$} (A3);  
  \path [line] (A1) --node[above]{$$} (A3);  
\end{tikzpicture}
\end{center}
A fact is that the computational effort in sequential eliminations is controlled by the sparsity of the subgraphs that appear in the process. The problem is that there is no good way of finding the right elimination order. What is worse is that, by the recent results in \cite{kwi10}, it is even hard to say if such an order exists or not. The difficulty of inference in a DAG is controlled by the so called tree-widths of its underlying graph. This is hard to compute for large graphs, and it thus makes sense to settle for an easier question: {\em what are some necessary conditions for a good elimination order to exist?}
 
One observation is that having small tree-width implies that, under some suitable ordering, the graphs that appear in the elimination process have small tree-width as well, and one can take this as a measure of efficiency for inference. In the above example, we see that the marginals of interest in $G$ are supported on the (unlabeled) Markov chain:

\begin{center}
\tikzset{node distance=1.5cm, auto}
\begin{tikzpicture}[auto]
  \node (A0){$M:$};
  \node (A1)[right of=A0,node distance=0.5cm]{$\bullet$};
  \node (A2)[right of=A1,node distance=1.5cm] {$\bullet$};
  \node (A3) [right of=A2,node distance=1.5cm]{$\bullet$};
  \draw [line] (A1) --node[above]{$$} (A2);  
  \path [line] (A2) --node[above]{$$} (A3);  
\end{tikzpicture}
\end{center}

However, the same cannot be said about the complete graph on four nodes. It is thus useful to know how simple the space of models that support the marginals of a certain DAG can be.  This gives us a notion of complexity of marginals in DAGs and it is clear that this notion depends on the class of the model, and not on how the model is represented as a graph. We formalize this notion as follows:

\begin{definition} \begin{enumerate}[(a)] \item Suppose that the nodes in $M$ are a subset of the nodes in $G$. We say that $M$ lies below $G$ if 
\[
\forall Q_{YX}\in \mLoc(G) \quad\exists\quad P_{X}\in \mLoc(M)\quad s.t.\quad Q_{YX}(\mathcal{A})=\int_\mathcal{A} K_{Y|X}dP_X
\]
Likewise, we say that $G$ lies above $M$. 
\item $M$ is said to be minimal w.r.t the above property if 
\[
\mLoc(M')\not\subset \mLoc(M)
\]
for all $M'$ that satisfy (a).
\item The class $[M]$ is said to lie below $[G]$ if (a) holds after some relabeling of variables (in $G$ or $M$). 
\end{enumerate}
\end{definition}

In other words, $M$ lies below $G$ if the marginal of any $G$-compatible distribution on the subspace associated to $M$ factorizes w.r.t it. We note that the elimination of a DAG need not be a DAG (see \cite{ric02}), but in cases that it is, the notions of minimal model and elimination model coincide. 

The main question of interest is to decide when $[M]$ can lie below $[G]$. 
It is clear that an enumerative approach to answer this question is problematic, even for a small graph $M$, as the number of possibilities grow exponentially with the size of the eliminated subset (which grows as $M$ gets smaller). The next proposition gives a necessary condition. 

\begin{proposition}
Let $M,G$ be DAGs and suppose $e:=|V_G|-|V_M|\ge 0$. Then if $[M]$ lies below $[G]$, there exists an embedding\footnote{By an embedding we mean a closed immersion.} $\X_G\hookrightarrow \X_{M^e}$. 
\end{proposition}
\begin{proof}
If $M$ lies below $G$, then its implied relations are satisfied by the marginals in $\mLoc(G)$. This gives an inclusion $\mloc(G)\subset \mloc(M^e)$, which induces an embedding $\hat{X}_G\hookrightarrow \hat{X}_M$. This embedding extends to the projective closures.  
\end{proof}

\begin{example}
Let $G$ be the complete DAG on $n$ nodes and $M$ be a Markov chain on $m\le n$ nodes. It is clear, from dimension considerations, that $[M]$ cannot appear below $[G]$.
\end{example}

We use the following proposition later:
\begin{proposition}
If there exists an embedding of $\mathbb{P}^n\times\mathbb{P}^1$ into $\mathbb{P}^{m}$, then $m\ge 2n+1$.
\end{proposition}
\begin{proof}
This is perhaps easiest to see from Bezout's theorem, which states any two closed sub-varieties of $\mathbb{P}^m$ of complementary dimension must have a non-empty intersection. In $\mathbb{P}^n\times\mathbb{P}^1$, there are $n$-dimensional hyperplanes ${\mathbb{P}^n\times \{pt\}}$ that do no intersect, and this remains true after embedding. Thus each hyperplane must have codimension at least $n+1$ in the ambient projective space. 
\end{proof}
This proposition shows that when $m=1$ the Segre embedding $\mathbb{P}^n\times\mathbb{P}^m\hookrightarrow \mathbb{P}^{nm+n+m+1}$
uses the minimal target dimension{\footnote{In the general case, there are embeddings $\mathbb{P}^n\times\mathbb{P}^m\hookrightarrow \mathbb{P}^{2(m+n)-1}$ and these can be shown to have the smallest target dimension.}}. This can be useful for our purposes.

\begin{example}
Let $G_n$ be the following family of DAGs
\begin{center}
\tikzset{node distance=1.5cm, auto}
\begin{tikzpicture}
  \node (A0){$1$};
  \node (A1) [right of=A0,node distance=1.5cm]{$$};
  \node (B0) [below of=A0,node distance=1.5cm] {$2$};
  \node (B1) [below of=A1,node distance=1.5cm]{$n+1$};
  \node (B2) [right of=B1,node distance=1.5cm]{$n+2$};
  \node (C0) [below of=B0,node distance=.5cm]{$\vdots$};
  \node (D0) [below of=C0,node distance=1.25cm]{$n$};

  \path [line] (A0) --node[above]{$$} (B1);  
  \path [line] (B0) --node[above]{$$} (B1);  
  \path [line] (D0) --node[above]{$$} (B1);  
  \path [line] (B1) --node[above]{$$} (B2);  
\end{tikzpicture}
\end{center}
and $M_n$ be a disconnected V-structure on $n+1$ nodes
\begin{center}
\tikzset{node distance=1.5cm, auto}
\begin{tikzpicture}
  \node (A0){$1$};
  \node (A1) [right of=A0,node distance=1.5cm]{$$};
  \node (B0) [below of=A0,node distance=1.5cm] {$2$};
  \node (B1) [below of=A1,node distance=1.5cm]{$n+1$};
  \node (C0) [below of=B0,node distance=.5cm]{$\vdots$};
  \node (D0) [below of=C0,node distance=1.25cm]{$n$};

  \path [line] (B0) --node[above]{$$} (B1);  
  \path [line] (D0) --node[above]{$$} (B1);  
\end{tikzpicture}
\end{center}

One can check that $\X_{G_n}\simeq\mathbb{P}^{n}\times\mathbb{P}^1$ and that $\X_{M_n^e}\simeq\mathbb{P}^{2n}$. Then $[M_n]$ does not lie below $[G_n]$ by the above proposition. This example is tight, in the sense that $[M_{n-1}]$ does lie below $[G_n]$ as eliminating the nodes $n+1$ and $n+2$ makes $G_n$ completely disconnected. In particular, $G_2$ does not lie above $\bullet\quad \bullet\to \bullet$.  
\end{example}

The examples above are special in that they all involve smooth projective varieties, and this is hardly an attribute of projective DAG varieties. Additional effort will be needed to deal with singularities. 
 
\begin{example}
Consider the DAG
\begin{center}
\tikzset{node distance=1.5cm, auto}
\begin{tikzpicture}[auto]
  \node (A0){$$};
  \node (A1)[right of=A0,node distance=1.5cm] {$$};
  \node (A3) [below of=A0,node distance=0.5cm]{$3$};
  \node (A3l) [left of=A3,node distance=1cm]{$$};
  \node (A3ul) [above of=A3l,node distance=0.75cm]{$1$};
  \node (A3dl) [below of=A3l,node distance=0.75cm]{$2$};
  \node (A4) [right of=A3,node distance=1.5cm]{$4$};
  \node (A5p) [right of=A4,node distance=1.5cm]{$$};
  \path [line] (A3) --node[above]{$$} (A4);  
  \path [line] (A3ul) --node[above]{$$} (A3);  
  \path [line] (A3dl) --node[above]{$$} (A3);  
  \path [line] (A3ul) --node[above]{$$} (A3dl);  
\end{tikzpicture}
\end{center}
We want to show that for all $n$, $[G]$ does not lie above $[M]$ where $[M]$ is the class of a V-structure $\bullet \rightarrow \bullet \leftarrow \bullet$. The extension $[M^e]$ of a V-structure to four nodes is the DAG
\begin{center}
\tikzset{node distance=1.5cm, auto}
\begin{tikzpicture}[auto]
  \node (A0){$$};
  \node (A1)[right of=A0,node distance=1.5cm] {$$};
  \node (A3) [below of=A0,node distance=0.5cm]{$\bullet$};
  \node (A3l) [left of=A3,node distance=1cm]{$$};
  \node (A3ul) [above of=A3l,node distance=0.75cm]{$\bullet$};
  \node (A3dl) [below of=A3l,node distance=0.75cm]{$\bullet$};
  \node (A4) [right of=A3,node distance=1.5cm]{$\bullet$};
  \node (A5p) [right of=A4,node distance=1.5cm]{$$};
  \path [line] (A3) --node[above]{$$} (A4);  
  \path [line] (A3ul) --node[above]{$$} (A3);  
  \path [line] (A3ul) --node[above]{$$} (A4);  
  \path [line] (A3dl) --node[above]{$$} (A4);  
  \path [line] (A3dl) --node[above]{$$} (A3);  
\end{tikzpicture}
\end{center}
Suppose that $[M]$ is below $[G]$ to obtain an embedding $\iota:\X_G\hookrightarrow \X_{M^e}$. We note that $\Pi_G$ is smooth everywhere except at the vertex of the affine cone over $\mathbb{P}^1\times\mathbb{P}^1$. Blowing up $\X_{M^e}\simeq\mathbb{P}^5$ at $\iota(p)$ gives a commutative diagram:
\begin{center}
\tikzset{node distance=1.5cm, auto}
\begin{tikzpicture}[auto]
  \node (A0){$\Bl_p(\X_G)$};
  \node (A0r)[right of=A0,node distance=2.5cm] {$\Bl_{\iota(p)}(\X_{M^e})$};
  \node (A1) [below of=A0,node distance=1.5cm]{$\X_G$};
  \node (A1r) [right of=A1,node distance=2.5cm]{$\X_{M^e}$};
  \path [line] (A0) --node[above]{$$} (A0r);  
  \path [line] (A0r) --node[above]{$$} (A1r);  
  \path [line] (A1) --node[above]{$$} (A1r);  
  \path [line] (A0) --node[above]{$$} (A1);  
\end{tikzpicture}
\end{center}
Since the property of being a closed immersion is stable under base change, it follows that the map $\Bl_p(\X_G)\to \Bl_{\iota(p)}(\mathbb{P}^5)$ is an embedding (see corollary II.7.15 in \cite{har77})
We note that $\Bl_p(\Pi_G)$ is smooth, and has the structure of a $\mathbb{P}^1$-bundle over $\mathbb{P}^2\times \mathbb{P}^1$. Likewise, the blow-up $\Bl_{\iota(p)}(\mathbb{P}^5)$ is a $\mathbb{P}^1$-bundle over $\mathbb{P}^4$. In other words, after blow ups, we get projective bundles over DAG varieties
\begin{center}
\tikzset{node distance=1.5cm, auto}
\begin{tikzpicture}[auto]
  \node (A0){$\Bl_p(\X_G)$};
  \node (A0r)[right of=A0,node distance=2.5cm] {$\Bl_{\iota(p)}(\X_{M^e})$};
  \node (A1) [below of=A0,node distance=1.5cm]{$\X_{G'}$};
  \node (A1r) [right of=A1,node distance=2.5cm]{$\X_{M'}$};
  \path [line] (A0) --node[above]{$$} (A0r);  
  \path [line] (A0r) --node[right]{$\mathbb{P}^1$-bundle} (A1r);  
  \path [line] (A1) --node[right]{$$} (A1r);  
  \path [line] (A0) --node[left]{$\mathbb{P}^1$-bundle} (A1);  
\end{tikzpicture}
\end{center}
where $G'$ looks like
\begin{center}
\tikzset{node distance=1.5cm, auto}
\begin{tikzpicture}[auto]
  \node (A0){$$};
  \node (A1)[right of=A0,node distance=1.5cm] {$$};
  \node (A3) [below of=A0,node distance=0.5cm]{$\bullet$};
  \node (A3l) [left of=A3,node distance=1cm]{$$};
  \node (A3ul) [above of=A3l,node distance=0.75cm]{$\bullet$};
  \node (A3dl) [below of=A3l,node distance=0.75cm]{$\bullet$};
  \node (A4) [right of=A3,node distance=1.5cm]{$\bullet$};
  \node (A5p) [right of=A4,node distance=1.5cm]{$$};
  \path [line] (A3) --node[above]{$$} (A4);  
  \path [line] (A3ul) --node[above]{$$} (A3);  
  \path [line] (A3dl) --node[above]{$$} (A3);  
\end{tikzpicture}
\end{center}
and $M'$ is
\begin{center}
\tikzset{node distance=1.5cm, auto}
\begin{tikzpicture}[auto]
  \node (A0){$$};
  \node (A1)[right of=A0,node distance=1.5cm] {$$};
  \node (A3) [below of=A0,node distance=0.5cm]{$\bullet$};
  \node (A3l) [left of=A3,node distance=1cm]{$$};
  \node (A3ul) [above of=A3l,node distance=0.75cm]{$\bullet$};
  \node (A3dl) [below of=A3l,node distance=0.75cm]{$\bullet$};
  \node (A4) [right of=A3,node distance=1.5cm]{$\bullet$};
  \node (A5p) [right of=A4,node distance=1.5cm]{$$};
  \path [line] (A3) --node[above]{$$} (A4);  
  \path [line] (A3ul) --node[above]{$$} (A4);  
  \path [line] (A3dl) --node[above]{$$} (A4);  
  \path [line] (A3dl) --node[above]{$$} (A3);  
\end{tikzpicture}
\end{center}
One can check that $\mathrm{Pic}(\Bl_{\iota(p)}(\X_{G}))=\mathbb{Z}^3$ while $\mathrm{Pic}(\Bl_p(\X_{M^e}))=\mathbb{Z}^2$. If an embedding exists, $\Bl(\X_{G})$ must be isomorphic to a smooth codimension one subvariety in $\Bl_p(\X_{M^e})$, which has Picard group isomorphic to $\mathbb{Z}^2$ by the Lefschetz's hyperplane theorem. This example shows that if two Markov chain relations fit in a DAG, they cannot combine to produce a pure independence relation. The fact that the relations fit in a DAG is essential here. 
\end{example}

It is common to encounter DAG varieties that do not have a simple geometric description as in the above families. One can work with numeric invariants in this case. 

\begin{example}
Let us check if the DAG
\begin{center}
\tikzset{node distance=1.5cm, auto}
\begin{tikzpicture}[auto]
  \node (A0){$$};
  \node (A1)[right of=A0,node distance=1.5cm] {$$};
  \node (A3) [below of=A0,node distance=0.5cm]{$4$};
  \node (A3l) [left of=A3,node distance=1cm]{$$};
  \node (A3ul) [above of=A3l,node distance=0.75cm]{$$};
  \node (A3dl) [below of=A3l,node distance=0.75cm]{$1$};
  \node (S) [left of=A3dl,node distance=0.5cm]{$G:$};
  \node (A3r) [right of=A3,node distance=1.5cm]{$2$};
  \node (A4) [below of=A3r,node distance=.75cm]{$5$};
  \node (A6) [right of=A4,node distance=1.5cm]{$6$};
  \node (A5p) [right of=A4,node distance=1.5cm]{$$};
  \node (B4) [below of=A3,node distance=1.5cm]{$3$};
  \path [line] (A3) --node[above]{$$} (A4);  
  \path [line] (A3r) --node[above]{$$} (A3);  
  \path [line] (A3dl) --node[above]{$$} (A3);  
  \path [line] (A3dl) --node[above]{$$} (B4);  
  \path [line] (B4) --node[above]{$$} (A3);  
  \path [line] (A4) --node[above]{$$} (A6);  
  \path [line] (B4) --node[above]{$$} (A4);  
  \path [line] (A3r) --node[above]{$$} (A6);  

\end{tikzpicture}
\end{center}
lies above $M:\bullet\to\bullet\to\bullet$. There is an embedding $\iota:\X_G\hookrightarrow \mathbb{P}^{15}$ defined by the ideal sheaf $\mathcal{I}_G$. If $[G]$ lies above $[M]$, then there must exist an irreducible quadric hyper-surface in $\mathbb{P}^{15}$ that contains the image of $\iota$. This means that the map $\Gamma(\mathbb{P}^{15},\mathcal{O}_{\mathbb{P}^{15}}(2))\to \Gamma(\mathbb{P}^{15},\iota_*\mathcal{O}_{\X_G}(2))$ has an irreducible quadratic polynomial in its kernel, which corresponds to a global section of $\mathcal{I}_G(2)$. Using Magma, we compute the Hilbert polynomial of $\mathcal{I}_G$ and 
conclude that $h^0\mathcal{I}_G(1)=2$ and  $h^0\mathcal{I}_G(2)=31$. Thus there are 2 linearly independent hyperplanes\footnote{Two hyperplanes are said to be linearly independent if they are defined by linearly independent sections of $O(1)$.} that contain the image of $\iota$. These in turn give 31 reducible (and thus no irreducible) quadrics containing the image. This shows that $[G]$ does not lie above $[M]$.  
\end{example}

\begin{proposition}
Let $M$ be a Markov chain on three nodes and $G$ be any DAG on $n$ nodes. Set $a_i:=h^0(\mathcal{I}_G(i))$.  Then $[G]$ lies above $[M]$ if and only if
\[
a_2-\frac{a_1}{2}(n^2-n-a_1+3)>0.
\]
\end{proposition}
This section justifies the following questions: 1) {\em what are some invariants of DAG varieties that are useful for ruling out embeddings of the above type?} 2) {\em how are such invariants related to the combinatorial data of the DAG?} 3) {\em when can the numbers $a_i$ be read from the Hilbert polynomial of $\mathcal{I}_G$?} We plan to come back to these questions in a future paper. 
\\

\bibliographystyle{IEEEtran}

\bibliography{IEEEabrv,biblio}

\end{document}